\documentclass[10pt]{amsart}
\usepackage{amsmath,amssymb,amsthm,amsfonts,amsopn,cite,mathrsfs}

\usepackage{mathptmx}       
\usepackage{helvet}         
\usepackage{courier}        
\usepackage{type1cm}        
\usepackage{graphicx}        
\usepackage[position=top,aboveskip=0pt,captionskip=4pt]{subfig}			 
\usepackage{multicol}        
\usepackage[bottom]{footmisc}
\usepackage[hyphens]{url}
\usepackage{algorithm}
\usepackage{algpseudocode}

\usepackage[normalem]{ulem}

\DeclareMathAlphabet{\mathbfit}{OML}{cmm}{b}{it}

\makeindex             
                       
\newcommand{\bd}{\mathbf}
\newcommand{\bdi}{\mathbfit}
\def\lcm{\operatorname{lcm}}
\def\RR{\mathbb{R}}
    \def\CC{\mathbb{C}}
    \def\ZZ{\mathbb{Z}}
    
    \def\NN{\mathbb{N}}
    \def\TT{\mathbb{T}}
    
\newcommand{\Hil}[0]{
\mathcal{H} 
}
\newcommand{\norm}[2]{
\left\| #2 \right\|_{#1}
}

\def\<{\langle}
\def\>{\rangle}
\def\SSphi{\bd S_\Phi}

\newcommand{\dt}{$\Delta T$}
\newcommand{\dx}{$\Delta \xi$}

\theoremstyle{plain}
\newtheorem{Thm}{Theorem}
\newtheorem{Def}{Definition}

 \newtheorem{Cor}{Corollary}
 
 \newtheorem{Pro}{Proposition}
 \theoremstyle{definition}

\begin{document}

\title{Frame Theory for Signal Processing in Psychoacoustics}
\author{Peter Balazs, Nicki Holighaus, Thibaud Necciari, and Diana Stoeva}
\address{Acoustics Research Institute, Austrian Academy of Sciences, Wohllebengasse 12-14, 1040 Wien, Austria, e-mail: 
peter.balazs@oeaw.ac.at, 
nicki.holighaus@oeaw.ac.at, thibaud.necciari@oeaw.ac.at,\linebreak
diana.stoeva@oeaw.ac.at
}
%
%
\begin{abstract}
This review chapter
aims to strengthen the link between frame theory and signal processing tasks in psychoacoustics.
On the one side, the basic concepts of frame theory are presented and some proofs are provided to explain those concepts in some detail. 
The goal is to reveal to hearing scientists how this mathematical theory could be relevant for their research. In particular, we focus on frame theory in a filter bank approach, which is probably the most relevant view-point for audio signal processing. On the other side, basic psychoacoustic concepts are presented to stimulate mathematicians to apply their knowledge in this field.\end{abstract}

\maketitle

\section{Introduction}\label{Sec:intro0}

In the fields of audio signal processing and hearing research, continuous research efforts are dedicated to the development of optimal representations of sound signals, suited for particular applications. 
However, each application and each of these two disciplines has specific requirements with respect to \textit{optimality} of the transform. 

For researchers in audio signal processing, an optimal signal representation should allow to extract, process, and re-synthesize relevant information, and avoid any useless inflation of the data, while at the same time being easily interpretable. In addition, although not a formal requirement, but being motivated by the fact that most audio signals are targeted at humans, the representation should take human auditory perception into account.  Common tools used in signal processing are linear time-frequency analysis methods that are mostly implemented as filter banks.

For hearing scientists, an optimal signal representation should allow to extract the perceptually relevant information in order to better understand sound perception. In other terms, the representation should reflect the  peripheral ``internal''  representation of sounds in the human auditory system. The tools used in hearing research are computational models of the auditory system. Those models come in various flavors but their initial steps in the analysis process usually consist in several parallel bandpass filters followed by one or more nonlinear and signal-dependent processing stages. The first stage, implemented as a (linear) filter bank, aims to account for the spectro-temporal analysis performed in the cochlea.
 The subsequent nonlinear stages aim to account for the various nonlinearities that occur in the periphery (e.g. cochlear compression) and at more central processing stages of the nervous system (e.g. neural adaptation). A popular auditory model, for instance, is the compressive gammachirp filter bank (see Sec.~\ref{ssec:audfilters}). In this model, a linear prototype filter is followed by a nonlinear and level-dependent compensation filter to account for cochlear compression. Because auditory models are mostly intended as perceptual analysis tools, they do not feature a synthesis stage, i.e. they are not necessarily invertible. Note that a few models do allow for an approximate reconstruction, though.

It becomes clear that filter banks play a central role in hearing research and audio signal processing alike, although the requirements of the two disciplines differ. 
This divergence of the requirements, in particular the need for signal-dependent nonlinear processing in auditory models, may contrast with the needs of signal processing applications. 
But even within each of those fields, demands for the properties of transforms are diverse, as becoming evident by the many already existing methods.
Therefore, it can be expected that the perfect signal representation, i.e. one that would have all desired properties for arbitrary applications in one or even both fields, does not exist.

This manuscript demonstrates how \emph{frame theory} can be considered a particularly useful \emph{conceptual} background for scientists in both hearing and audio processing, and presents some first motivating applications. 
Frames provide the following general properties: \emph{perfect reconstruction}, \emph{stability}, \emph{redundancy}, and a \emph{signal-independent, linear inversion procedure}. In particular, frame theory can be used to analyze any filter bank, thereby providing useful insight into its structure and properties. In practice, if a filter bank construction (i.e. including both the analysis and synthesis filter banks) satisfies the frame condition (see Sec.~\ref{sec:erbfb}), it benefits from all the frame properties mentioned above. Why are those properties essential to researchers in audio signal processing and hearing science?

\textbf{Perfect reconstruction property:} With the possible exception of frequencies outside the audible range, a non-adaptive analysis filter bank {, i.e. one that is general, not signal-dependent,} has no means of determining and extracting exactly the perceptually relevant information. 
For such an extraction, signal-dependent information would be crucial. 
Therefore, the only way to ensure that a linear, signal-independent analysis stage\footnote{As given by any fixed analysis filter bank.}, possibly followed by a nonlinear processing stage, captures all \emph{perceptually relevant signal components} is to ensure that it does \emph{not lose any} information at all. 
This, in fact, is \emph{equivalent to being perfectly invertible}, i.e. having a perfect reconstruction property. 
Thus, this property benefits the user even when reconstruction is not intended per-se. Note that in general ``being perfectly invertible'' need not necessarily imply that a concrete inversion procedure is known. In the frame case, a constructive method exists, though.

\textbf{Stability:} 
For sound processing, 
 stability is essential in the sense that, for the analysis stage, when two signals are similar (i.e., their difference is small), the difference between their corresponding analysis coefficients should also be small. 
For the synthesis stage, a signal reconstructed from slightly distorted coefficients should be relatively close to the original signal, that is the one reconstructed from undistorted coefficients. 
 From an energy point of view, 
 signals which are similar in energy should provide analysis coefficients whose energy is also similar.
So the respective energies remain roughly proportional. In particular, considering a signal mixture, the combination of stability and linearity ensures that every signal component is represented and weighted according to its original energy. In other terms, individual signal components are represented proportional to their energy, which is very important for, e.g., visualization. Even in a perceptual analysis, where inaudible components should not be visualized equally to audible components having the same energy, this stability property is important. To illustrate this, recall that the nonlinear post-processing stages in auditory models are signal dependent. 
That is, also the inaudible information can be essential to properly characterize the nonlinearity. For instance, consider again the setup of the \emph{compressive gammachirp} model where an intermediate representation is obtained through the application of a linear analysis filter bank  to the input signal. 
The result of this linear transform determines the shape of the subsequent nonlinear compensation filter. Note that the \emph{whole} intermediate representation is used. Consequently, the proper estimation of the nonlinearity crucially relies on the signal representation being accurate, i.e. \emph{all} signal components being represented and appropriately weighted. This \emph{accurateness} comes for free if the analysis filter bank forms a frame.

\textbf{Signal-independent, linear inversion:} 
A consistent (i.e. signal-independent) inversion procedure is of great benefit in signal processing applications. 
It implies that 
a single algorithm/implementation can perform all the necessary synthesis tasks. 
For nonlinear representations, finding a signal-independent procedure which provides a stable reconstruction is a highly nontrivial affair, if it is at all possible. With linear representations, such a procedure is easier to determine and this can be seen as an advantage of the linearity. 
The linearity provided by the reconstruction algorithm also significantly simplifies separation tasks. 
In a linear representation, a separation in the coefficient (time-frequency) domain, i.e. before synthesis, is equivalent to a separation in the signal domain. 
Such a property is highly relevant, for instance, to computational auditory scene analysis systems that, to some extent, are sound source separators (see Sec.~\ref{sec:casa0}).

\textbf{Redundancy:} 
 Representations which are sampled at critical density are often unsuitable for visualization, since they lead to a low resolution, which may lead to many distinct signal components being integrated into a single coefficient of the transform.
Thus, the individual coefficients may contain information from a lot of different sources, which makes them hard to interpret. Still, the whole set of coefficients captures all the desired signal information if (and only if) the transform is invertible.
Redundancy provides higher resolution and so components that are separated in time or in frequency can be separated in the transform domain. 
Furthermore, redundant representations are smoother and therefore easier to read than their critically sampled counterparts.

Moreover, redundant representations provide some resistance against noise and errors. This is in contrast to non-redundant systems, where distortions can not be compensated for. This is used for de-noising approaches. In particular, if a signal is synthesized in a straight-forward way from noisy (redundant) coefficients, the synthesis process has the tendency to reduce the energy of the noise, i.e. there is some noise cancellation.
\\

Besides the above properties, which are direct consequences of the frame inequalities, 
the generality of frame theory enables the consideration of \emph{additional important properties}. In the setting of perceptually motivated audio signal analysis and processing, these include:

\textbf{Perceptual relevance:} 
We have stressed that the only way to ensure that all perceptually relevant information is kept is to accurately capture all the information by using a stable and perfectly invertible system for analysis. However, in an auditory model or in perceptually motivated signal processing, perceptually irrelevant components should be discarded at some point. 
If only a linear signal processing framework is desired, this can be achieved by applying a perceptual weighting\footnote{Different frequency ranges are given varying importance in the auditory system} and a masking model, see Sec.~\ref{sec:psychotheory}. 
If a nonlinear auditory model like the compressive gammachirp filter bank is used, recall that the nonlinear stage is mostly determined by the coefficients at the output of the linear stage. Therefore, all information should be kept up to the nonlinear stage. 
In other words, discarding information already in the analysis stage might falsify the estimation of the nonlinear stage, 
thereby resulting in an incorrect perceptual analysis. 
We want to stress here the importance of being able to \emph{selectively} discard unnecessary information, in contrast to information being \emph{involuntarily lost} during the analysis and/or synthesis procedures.

\textbf{A flexible signal processing framework:} All stable and invertible filter banks 
 form a frame and therefore benefit from the frame properties discussed above. In addition, using filter banks that are frames allows for flexibility. For instance, one can gradually tune the signal representation such as the \emph{time-frequency resolution}, analysis filters' \emph{shape} and \emph{bandwidth}, \emph{frequency scale}, \emph{sampling density} etc., while at the same time retaining the crucial frame properties. 
It can be tremendously useful to provide a single and adaptable framework that allows to switch model parameters and/or transition between them. By staying in the common general setting of filter bank frames, the linear filter bank analysis in an auditory model or signal processing scheme can be seen as an exchangeable, practically self-contained block in the scheme. Thus, the filter bank parameters, e.g. those mentioned before, can be tuned by scientists according to their preference, without the need to redesign the remainder of the model/scheme.  
Such a common background leads to results being more comparable across research projects and thus benefits not only the individual researcher, but the whole field. Two main advantages of a common background are the following: first, the properties and parameters of various models can be easily interpreted and compared across contributions; second, by the adaption of a linear model to obtain a nonlinear model the new model parameters 
 remain interpretable.

\textbf{Ease of integration:} Filter banks are already a common tool in both hearing science and signal processing. Integrating a filter bank frame into an existing analysis/processing framework will often only require minor modifications of existing approaches. Thus, frames provide a theoretically sound foundation without the need to fundamentally re-design the remainder of your analysis (or processing) framework.\\

\emph{In some cases, you might already implicitly use frames without knowing it. In that case, we provide here the conceptual background necessary to unlock the full potential of your method.} \\

The rest of this chapter is organized as follows: In Section \ref{sec:psychotheory}, we provide basic information about the human auditory system and introduce some psychoacoustic concepts. 
In Section \ref{sec:frameth} we present the basics of frame theory providing the main definitions and a few crucial mathematical statements.
In Section \ref{sec:erbfb} we provide some details on filter bank frames.
 The chapter concludes with Section \ref{sec:appli} where some examples are given for the application of frame theory to signal processing in psychoacoustics.

\section{\label{sec:psychotheory}The auditory analysis of sounds}

This section provides a brief introduction to the human auditory system. Important concepts that are relevant to the problems treated in this chapter are then introduced, namely auditory filtering and auditory masking. For a more complete description of the hearing organ, the interested reader is referred to e.g.~\cite{Fastl:2006a,Plack:2013b}.

\subsection{\label{ssec:anatomy}Ear's anatomy}

The human ear is a very sensitive and complex organ whose function is to transform pressure variations in the air into the percept of sound. To do so, sound waves must be converted into a form interpretable by the brain, specifically into neural action potentials. Fig.~\ref{fig:earanatomy} shows a simplified view of the ear's anatomy. Incoming sound waves are guided by the pinna into the ear canal and cause the eardrum to vibrate. Eardrum vibrations are then transmitted to the cochlea by three tiny bones that constitute the ossicular chain: the malleus, incus, and stapes. The ossicular chain acts as an impedance matcher. Its function is to ensure efficient transmission of pressure variations in the air into pressure variations in the fluids present in the cochlea. The cochlea is the most important part of the auditory system because it is where pressure variations are converted into neural action potentials.

The cochlea is a rolled-up tube filled with fluids and divided along its length by two membranes, the Reissner's membrane and basilar membrane (BM). A schematic view of the unrolled cochlea is shown in Fig.~\ref{fig:earanatomy} (the Reissner's membrane is not represented). It is the response of the BM to pressure variations transmitted through the ossicular chain that is of primary importance. Because the mechanical properties of the BM vary across its lengths (precisely, there is a gradation of stiffness from base to apex), BM stimulation results in a complex movement of the membrane. In case of a sinusoidal stimulation, this movement is described as a traveling wave. The position of the peak in the pattern of vibration depends on the frequency of the stimulation. High-frequency sounds produce maximum displacement of the BM near the base with little movement on the rest of the membrane. Low-frequency sounds rather produce a pattern of vibration which extends all the way along the BM but reaches a maximum 
before the apex. The frequency that gives the maximum response at a particular point on the BM is called the ``characteristic frequency'' (CF) of that point. In case of a broadband stimulation (e.g. an impulsive sound like a click), all points on the BM will oscillate. In short, the BM separates out the spectral components of a sound similar to a Fourier analyzer.

The last step of peripheral processing is the conversion of BM vibrations into neural action potentials. This is achieved by the inner hair cells that sit on top of the BM. There are about 3500~inner hair cells along the length of the cochlea ($\approx$35~mm in humans). The tip of each cell is covered with sensor hairs called stereocilia. The base of each cell directly connects to auditory nerve fibers. When the BM vibrates, the stereocilia are set in motion, which results in a bio-electrical process in the inner hair cells and, finally, in the initiation of action potentials in auditory nerve fibers. Those action potentials are then coded in the auditory nerve and conveyed to the central system where they are further processed to end up in a sound percept. Because the response of auditory nerve fibers is also frequency specific and the action potentials vary over time, the ``internal representation'' of a sound signal in the auditory nerve can be likened to a time-frequency representation.

\begin{figure}[ht]
	\begin{center}
	  \includegraphics[width=.75\textwidth]{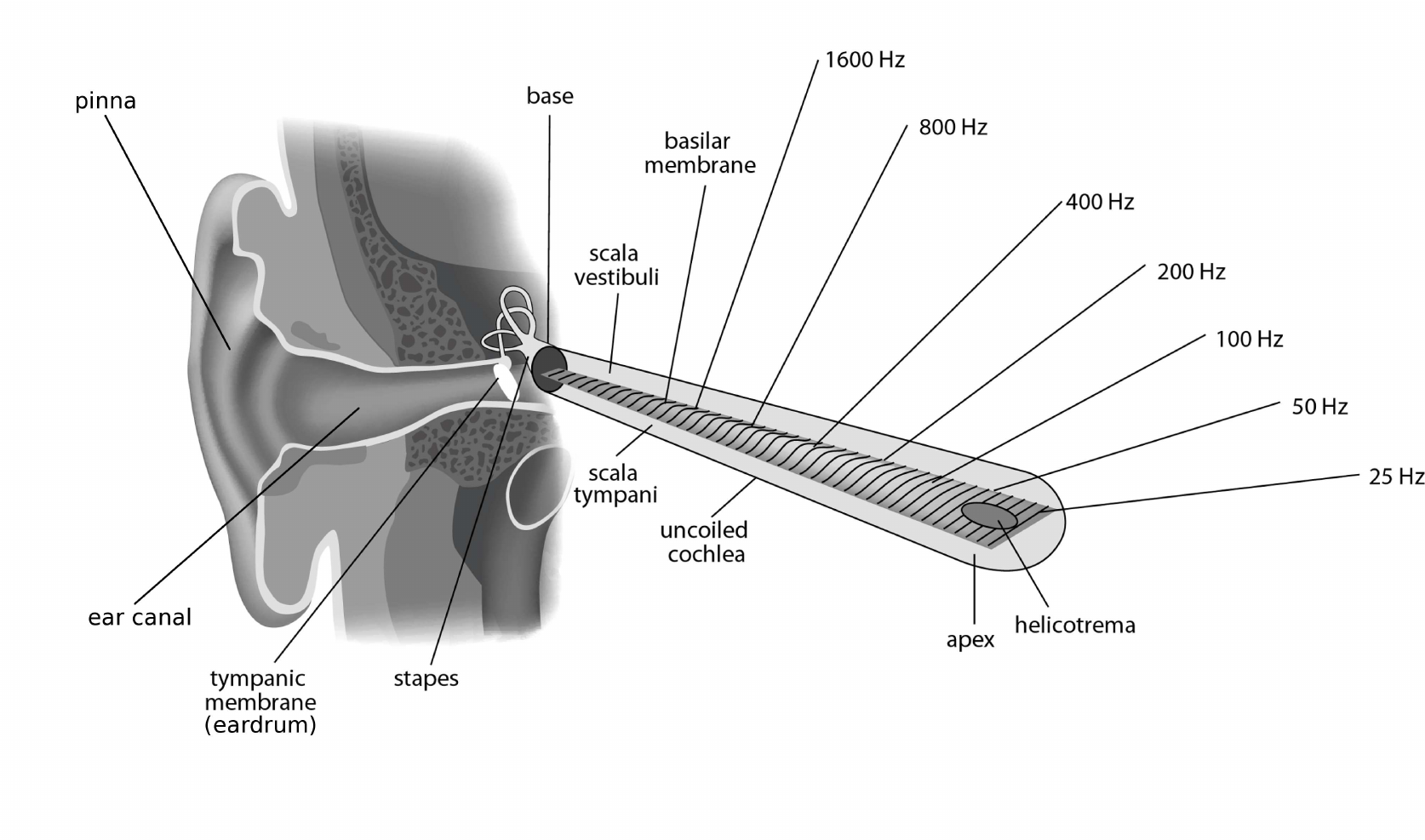}
	\end{center}
        \caption{Anatomy of the human ear with a schematic view of the unrolled cochlea. Adapted from~\cite{Kern:2008a}.}
	\label{fig:earanatomy}
\end{figure}

\subsection{\label{ssec:audfilters}The auditory filters concept}

Because of the frequency-to-place transformation (also called tonotopic organization) in the cochlea, and the transmission of time-dependent neural signals, the BM can be modeled in a first linear approximation as a bank of overlapping bandpass filters, named ``critical bands'' or ``auditory filters''. The center frequencies and bandwidth of the auditory filters, respectively, approximate the CF and width of excitation on the BM. Noteworthy, the width of excitation depends on level as well: patterns become wider and asymmetric as sound level increases (e.g.~\cite{Glasberg:1990a}).
Several auditory filter models have been proposed based on the results from psychoacoustics experiments on masking (see e.g.~\cite{Lyon:2010a} and Sec.~\ref{ssec:masking}). A popular auditory filter model is the gammatone filter \cite{Patterson:1992a} (see Fig~\ref{fig:gammatones}). Although gammatone filters 
do not capture the level dependency of the actual auditory filters, their ease of implementation made them popular in audio signal processing (e.g. \cite{Valero:2012a,Zhao:2012a}). More realistic auditory filter models are, for instance, the roex and gammachirp filters \cite{Glasberg:1990a,Unoki:2006a}. Other level-dependent and more complex auditory filter banks include for example the dual resonance non-linear filter bank \cite{Lopez-Poveda:2001a} or the dynamic compressive gammachirp filter bank \cite{Irino:2006b}. The two approaches in \cite{Lopez-Poveda:2001a,Irino:2006b} feature a linear filter bank followed by a signal-dependent nonlinear stage. As mentioned in the introduction, this is a particular way of describing a nonlinear system by modifying a linear system. Finally, it is worth noting that besides psychoacoustic-driven auditory models, mathematically founded models of the auditory periphery have been proposed. Those include, for instance, the wavelet auditory model \cite{bete94} or the 
``EarWig'' time-frequency distribution \cite{ODonovan:2005a}.

\begin{figure}
	\centering
	\includegraphics[scale=0.4]{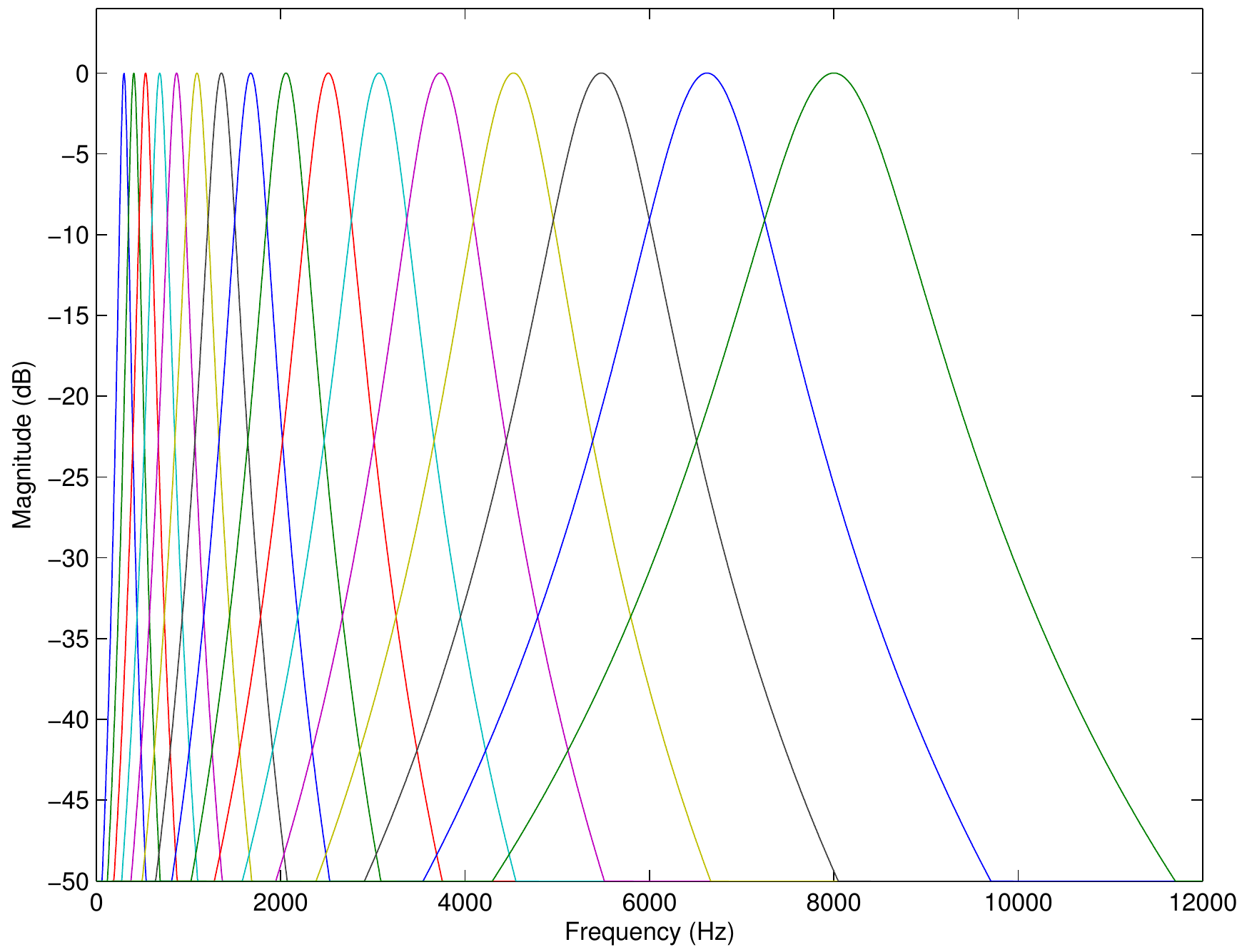}
	\caption{A popular auditory filter model: the gammatone filter bank. The magnitude responses (in~dB) of 16~gammatone filters in the frequency range 300-8000~Hz are represented on a linear frequency scale.}
	\label{fig:gammatones}
\end{figure}

The bandwidth of the auditory filters has been determined based on psychoacoustic experiments. The estimation of bandwidth based on loudness perception experiments gave rise to the concept of Bark bandwidth defined by~\cite{Zwicker:1980a}
\begin{equation}\label{eq:Barkbw}
	BW_{\mathrm{Bark}} = 25 + 75 \left( 1 + 1.4\times10^{-6} \xi^2\right)^{0.69}
\end{equation}

\noindent where $\xi$ denotes the frequency and $BW$ denotes the bandwidth, both in~Hz. Another popular concept is the equivalent rectangular bandwidth (ERB), that is the bandwidth of a rectangular filter having the same peak output and energy as the auditory filter. The estimations of ERBs are based on masking experiments. The ERB is given by~\cite{Glasberg:1990a}
\begin{equation}\label{eq:erbw}
	BW_{\mathrm{ERB}} = 24.7 + \frac{\xi}{9.265} \quad .
\end{equation}

\noindent $BW_{\mathrm{Bark}}$ and $BW_{\mathrm{ERB}}$ are commonly used in psychoacoustics and signal processing to approximate the auditory spectral resolution at low to moderate sound pressure levels (i.e. 30--70~dB) where the auditory filters' shape remains symmetric and constant. See for example \cite{Glasberg:1990a,Unoki:2006a} for the variation of $BW_{\mathrm{ERB}}$ with level.

Based on the concepts of Bark and ERB bandwidths, corresponding frequency scales have been proposed to represent and analyze data on a scale related to perception. To describe the different mappings between the linear frequency domain and the nonlinear perceptual domain we introduce the function $F_{\mathrm{AUD}}: \xi \rightarrow \mathrm{AUD}$ where $\mathrm{AUD}$ is an auditory unit that depends on the scale. The Bark scale is~\cite{Zwicker:1980a}
\begin{equation}\label{eq:Barkrate}
	F_\mathrm{Bark}(\xi) =  13 \arctan (0.00076 \xi) + 3.5 \arctan (\xi / 7500)^2
\end{equation}
and the ERB scale is~\cite{Glasberg:1990a}
\begin{equation}\label{eq:erbrate}
	F_\mathrm{ERB}(\xi) = 9.265 \ln \left( 1 + \frac{\xi}{228.8455}\right) \:.
\end{equation}
Both auditory scales are connected to the ear's anatomy. One AUD unit indeed corresponds to a constant distance along the BM. 1~$\mathrm{Bark}$ corresponds to 1.3~mm \cite{Fastl:2006a} while 1~$\mathrm{ERB}$ corresponds to 0.9~mm \cite{Glasberg:1990a,Greenwood:1990a}.

\subsection{\label{ssec:masking}Auditory masking}

The phenomenon of masking is highly related to the spectro-temporal resolution of the ear and has been the focus of many psychoacoustics studies over the last 70~years. Auditory masking refers to the increase in the detection threshold of a sound signal (referred to as the ``target'') due to the presence of another sound (the ``masker''). Masking is quantified by measuring the detection thresholds of the target in presence and absence of the masker; the difference in thresholds (in~dB) thus corresponds to the \textit{amount of masking}. In the literature, masking has been extensively investigated in the spectral or temporal domain. The results were used to develop models of spectral or temporal masking that are currently implemented in audio applications like perceptual coding (e.g. \cite{Painter:2000a,Ravelli:2008a}) or sound processing (e.g. \cite{Balazs:2010a,Gunawan:2010a}. Only a few studies investigated masking in the joint time-frequency domain. We present below some typical 
psychoacoustic results on spectral, temporal, and spectro-temporal masking. For more results and discussion on the origins of masking the interested reader is referred to e.g. \cite{Fastl:2006a,Necciari:2010a,Moore:2012a}.

In the following, we denote by $\xi_{\{M,T\}}$, $D_{\{M,T\}}$, and $L_{\{M,T\}}$ the frequency, duration, and level, respectively, of masker or target. Those signal parameters are fixed by the experimenter, i.e. they are known. The frequency shift between masker and target is $\Delta \xi = \xi_T - \xi_M$ and the time shift \dt~is defined as the onset delay between masker and target. Finally, $AM$ denotes the amount of masking in~dB.

\subsubsection{Spectral masking}\label{sec:SimMask0}

To study spectral masking, masker and target are presented simultaneously (since usually $D_M > D_T$, this is equivalent to saying that 0 $\leq \Delta T < D_M-D_T$) and \dx~is varied. There are two ways to vary \dx, either fix $\xi_T$ and vary $\xi_M$ or vice versa. Similarly, one can fix $L_M$ and vary $L_T$ or vice versa. In short, various types of masking curves can be obtained depending on the signal parameters. A common spectral masking curve is a masking pattern that represents $L_T$ or $AM$ as a function of $\xi_T$ or \dx$\,$ (see Fig.~\ref{fig:MPfreq}). To measure masking patterns, $\xi_M$ and $L_M$ are fixed and $AM$ is measured for various \dx. Under the assumption that $AM(\xi_T)$ corresponds to a certain ratio of masker-to-target energy at the output of the auditory filter centered at $\xi_T$, masking patterns measure the responses of the auditory filters centered at the individual $\xi_T$s. Thus, masking patterns can be used as indicator of the \textit{spectral spread of masking} of the 
masker or, in other terms, the spread of excitation of the masker on the BM. This spectral spread can in turn be used to derive a masking threshold, as used for example in audio codecs \cite{Painter:2000a}. See also Sec.~\ref{sec:Irrel0}.

Fig.~\ref{fig:MPfreq} shows typical masking patterns measured for narrow-band noise maskers of different levels ($L_M$ = 45, 65, and 85~dB SPL, as indicated by the different lines) and frequencies ($\xi_M$ = 0.25, 1, and 4~kHz, as indicated by the different vertical dashed lines). In this study, $D_M = D_T$ = 200~ms. The masker was a 80-Hz-wide band of Gaussian noise centered at $\xi_M$. The target was also a 80-Hz band of noise centered at $\xi_T$. The main properties to be observed here are:
\begin{enumerate}
\item[(i)] For a given masker (i.e. a pair of $\xi_M$ and $L_M$), $AM$ is maximum for \dx$\,$ = 0 and decreases as $|\Delta \xi|$ increases. This reflects the decay of masker excitation on the BM. 
\item[(ii)] Masking patterns broaden with increasing level. This reflects the broadening of auditory filters with increasing level \cite{Glasberg:1990a}. 
\item[(iii)] Masking patterns are broader at low than at high frequencies (see~\eqref{eq:Barkbw}-\eqref{eq:erbw}). This reflects the fact that the density of auditory filters is higher at low than at high frequencies. Consequently, a masker with a given bandwidth will excite more auditory filters at low frequencies. 
\end{enumerate}

\begin{figure}
	\centering
	\includegraphics[width=\textwidth]{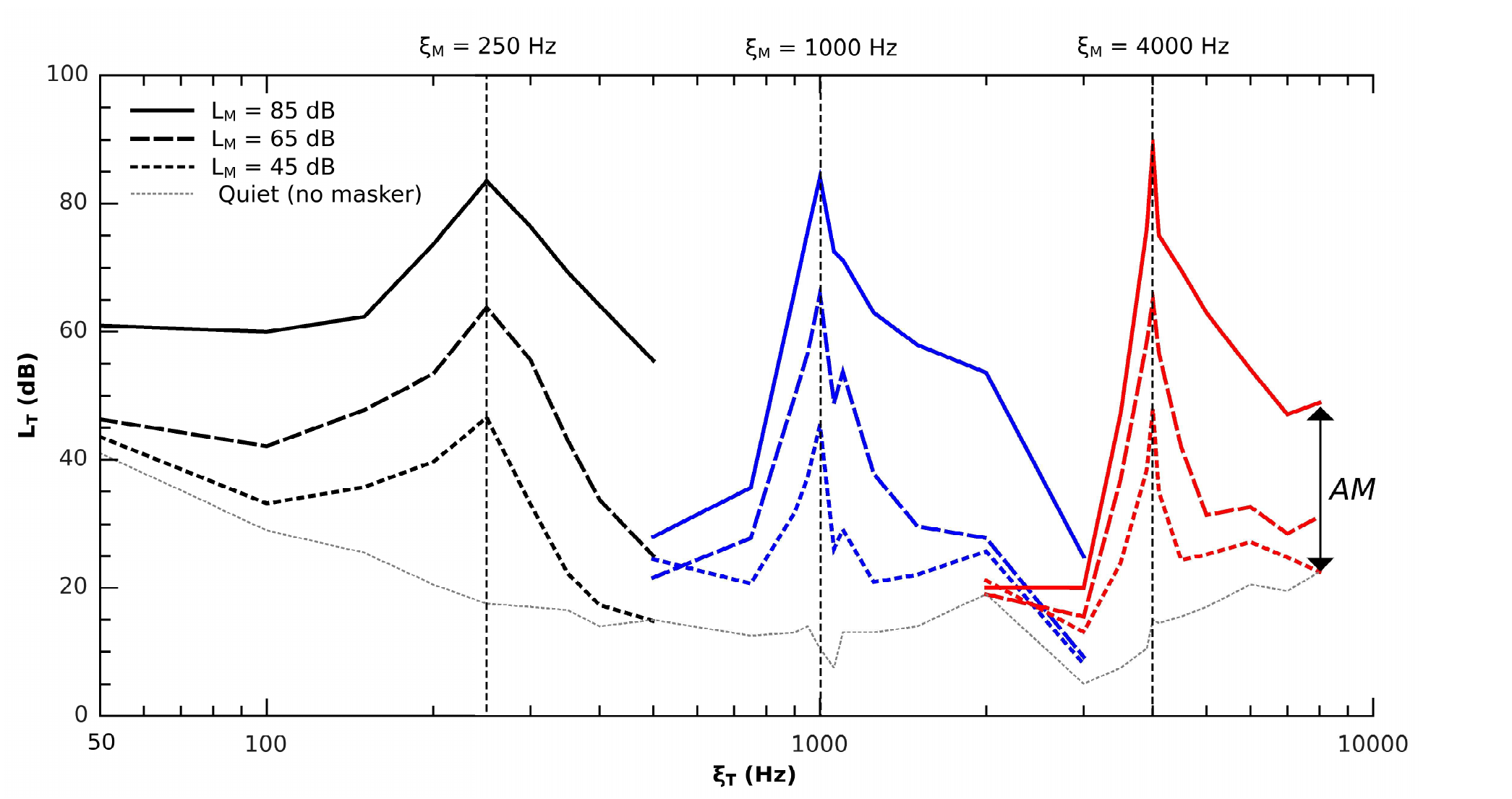}
	\caption{Masking patterns for narrow-band noise maskers of different levels and frequencies. $L_T$ (in dB SPL) is plotted as a function of $\xi_T$ (in Hz) on a logarithmic scale. The gray dotted curve indicates the threshold in quiet. The difference between any of the colored curves and the gray curve thus corresponds to $AM$, as indicated by the arrow. Source: mean data for listeners JA and AO in \cite[Experiment~3, Figs.~5-6]{Moore:1998a}.}
	\label{fig:MPfreq}
\end{figure}

\subsubsection{Temporal masking}

By analogy with spectral masking, temporal masking is measured by setting \dx~= 0 and varying \dt. \textit{Backward} masking is observed for $\Delta T <$ 0, that is when the target precedes the masker in time.  \textit{Forward} masking is observed for $\Delta T \geq D_M$, that is when the target follows the masker.
Backward masking is hardly observed for $\Delta T <$ -20~ms and is mainly thought to result from attentional effects \cite{Soderquist:1981a,Fastl:2006a}. 
In contrast, forward masking can be observed for $\Delta T \geq D_M$ + 200~ms. Therefore, in the following we focus on forward masking. 

Typical forward masking curves are represented in Fig.~\ref{fig:MPtime}. The left panel shows the effect of $L_M$ for $\xi_M = \xi_T$ = 4~kHz (mean data from \cite{Jesteadt:1982a}). In this study, masker and target were sinusoids ($D_M$ = 300~ms, $D_T$ = 20~ms). The main features to be observed here are (i) the temporal decay of forward masking is a linear function of $\log(\Delta T)$ and (ii) the rate of this decay strongly depends on $L_M$. The right panel shows the effect of $D_M$ for $\xi_T$ = 2~kHz and $L_M$ = 60~dB SPL (mean data from \cite{Zwicker:1984a}). In this study, the masker was a pulse of uniformly masking noise (i.e. a broad-band noise producing the same $AM$ at all frequencies in the range 0--20~kHz, see \cite{Fastl:2006a}). The target was a sinusoid with $D_T$ = 5~ms. It can be seen that the $AM$ (i.e. the difference between the connected symbols and the star) at a given \dt$\,$ increases with increasing $D_M$, at least for $\Delta T - D_M <$ 100~ms. Finally, a comparison of the two panels 
in Fig.~\ref{fig:MPtime} for $L_M$ = 60~dB indicates that, for $\Delta T - D_M \leq$ 50~ms, the 300-ms sinusoidal masker (empty diamonds left) produces more masking than the 200-ms broad-band noise masker (empty squares right). Despite the difference in $D_M$, increasing the duration of the noise masker to 300~ms is not expected to account for the difference in $AM$ of up to 20~dB observed here \cite{Fastl:2006a,Zwicker:1984a}. 

\begin{figure}
	\centering
	\includegraphics[width=\textwidth]{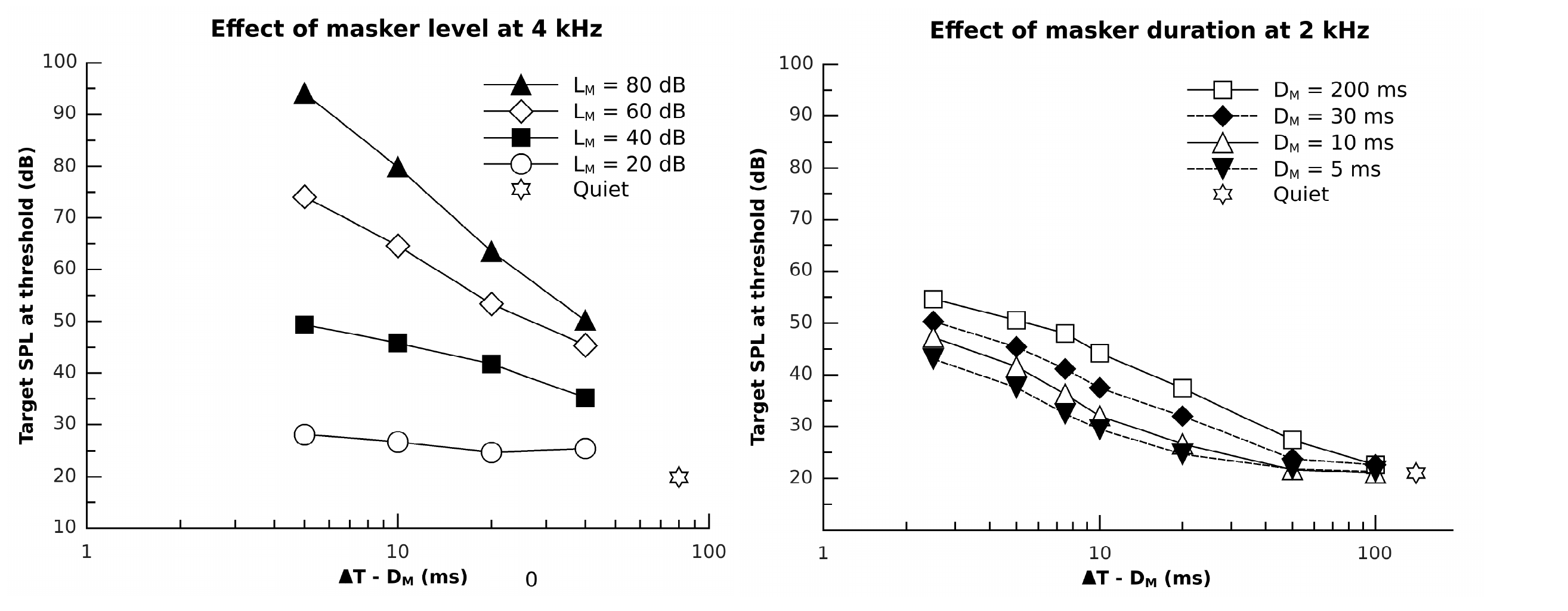}
	\caption{Temporal (forward) masking curves for sinusoidal (left) and broadband noise maskers (right). $L_T$ (in dB SPL) is plotted as a function of the temporal gap between masker offset and target onset, i.e. $\Delta T - D_M$ (in ms) on a logarithmic scale. Left panel: masking curves for various $L_M$s and $D_M$ = 300~ms (adapted from \cite{Jesteadt:1982a}). Right panel: masking curves for various $D_M$s and $L_M$ = 60~dB (adapted from \cite{Zwicker:1984a}). Stars indicate the target thresholds in quiet.}
	\label{fig:MPtime}
\end{figure}

\subsubsection{Time-frequency masking}\label{sec:tfmask0}
  
Only a few studies measured spectro-temporal masking patterns, that is \dt$\,$ and \dx$\,$ both systematically varied (e.g. \cite{Kidd:1981a,Soderquist:1981a}). Those studies mostly involved long ($D_M \geq$ 100~ms) sinusoidal maskers. In other words, those studies provide data on the time-frequency spread of masking for long and narrow-band maskers. In the context of time-frequency decompositions, a set of elementary functions, or ``atoms'', with good localization in the time-frequency domain (i.e. short and narrow-band) is usually chosen, see Sec.~\ref{sec:frameth}. To best predict masking in the time-frequency decompositions of sounds, it seems intuitive to have data on the time-frequency spread of masking for such elementary atoms, as this will provide a good match between the masking model and the sound decomposition. This has been investigated in \cite{Necciari:2010a}. Precisely, spectral, forward, and time-frequency masking have been measured using Gabor atoms of the form $s_i(t) = \sin (2 \pi \xi_
i t + \pi/4)e^{-\pi(\Gamma t)^2}$ with $\Gamma = 600~\mathrm{s}^{-1}$ as masker and target. According to the definition of Gabor atoms in~\eqref{Gabsystem}, the masker was defined by $s_M(t) = \Im\{e^{{\rm i}\pi/4}g_{\xi_M,0}\}$, where $\Im$ denotes the imaginary part, with a Gaussian window $\gamma(t)=e^{-\pi (\Gamma t)^2}$ and $\xi_M$ = 4~kHz. The masker level was fixed at $L_M$ = 80~dB. The target was defined by $s_T(t+\Delta T) = \Im\{e^{{\rm i}(\pi/4+2\pi \xi_T \Delta T)}\gamma_{\xi_T,-\Delta T}\}$ with $\xi_T = \xi_M + \Delta \xi$. The set of time-frequency conditions measured in \cite{Necciari:2010a} is illustrated in Fig.~\ref{fig:TFconds}. Because in this particular case we have $\xi_T \Delta T \in \mathbb{N}$, the target term reduces to $s_T(t+\Delta T) = \Im\{e^{{\rm i}(\pi/4)}\gamma_{\xi_T,-\Delta T}\}$. The mean masking data are summarized in Fig.~\ref{fig:MPTF}. These data, together with those collected by Laback et al on the additivity of spectral \cite{Laback:2013a} and temporal masking \cite{
Laback:2011a} for the same Gabor atoms, 
constitute a crucial basis for the development of an accurate time-frequency masking model to be used in audio applications like audio coding or audio processing (see Sec.~\ref{sec:appli}). 

\begin{figure}[!t]
	  \begin{center}
		 \subfloat[Experimental conditions]{%
		 	\label{fig:TFconds}
		 	\includegraphics[height=4cm]{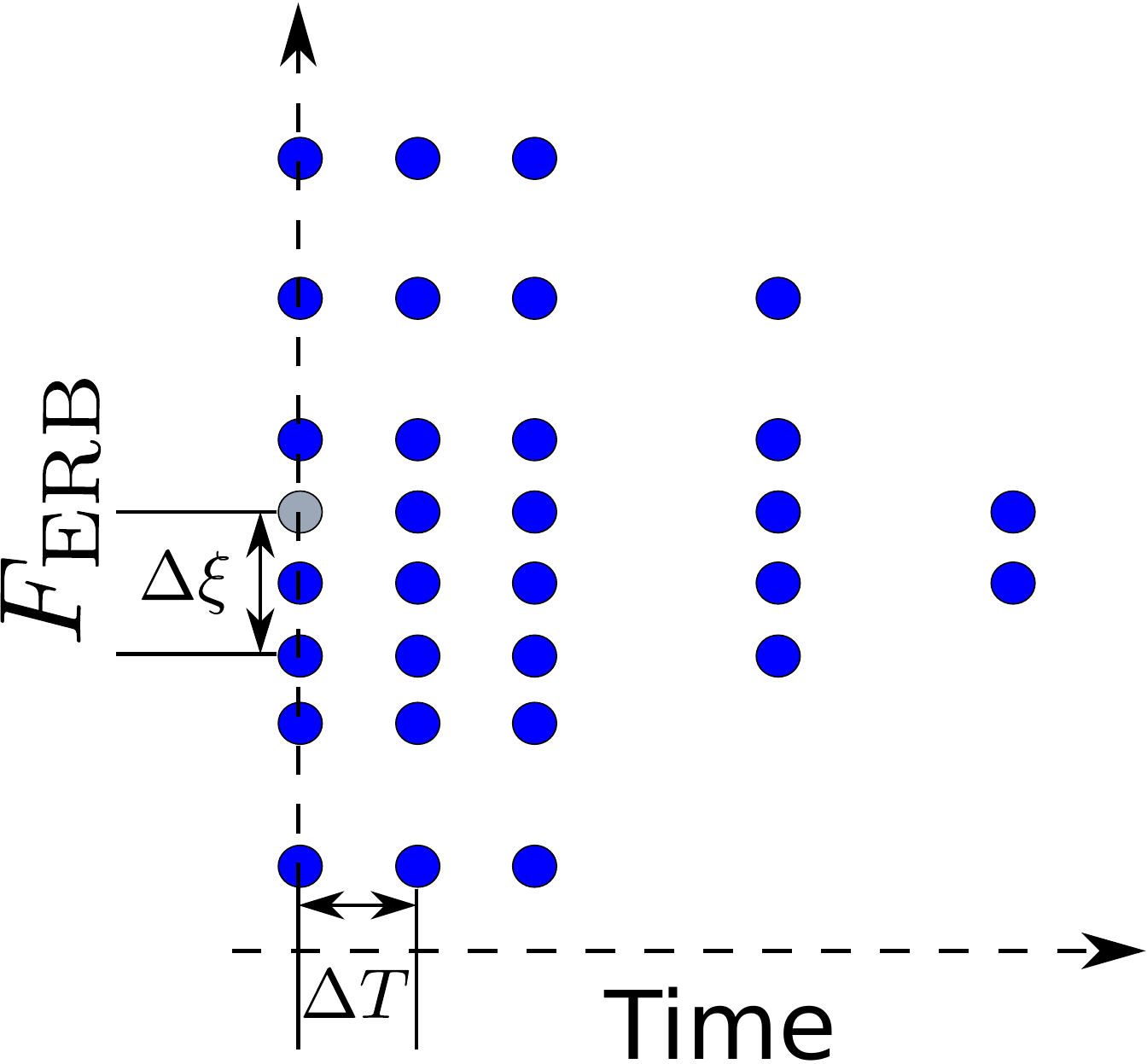}}
		 \subfloat[Mean results]{\vspace{-5pt}%
		 	\label{fig:MPTF}
		 	\includegraphics[height=5cm]{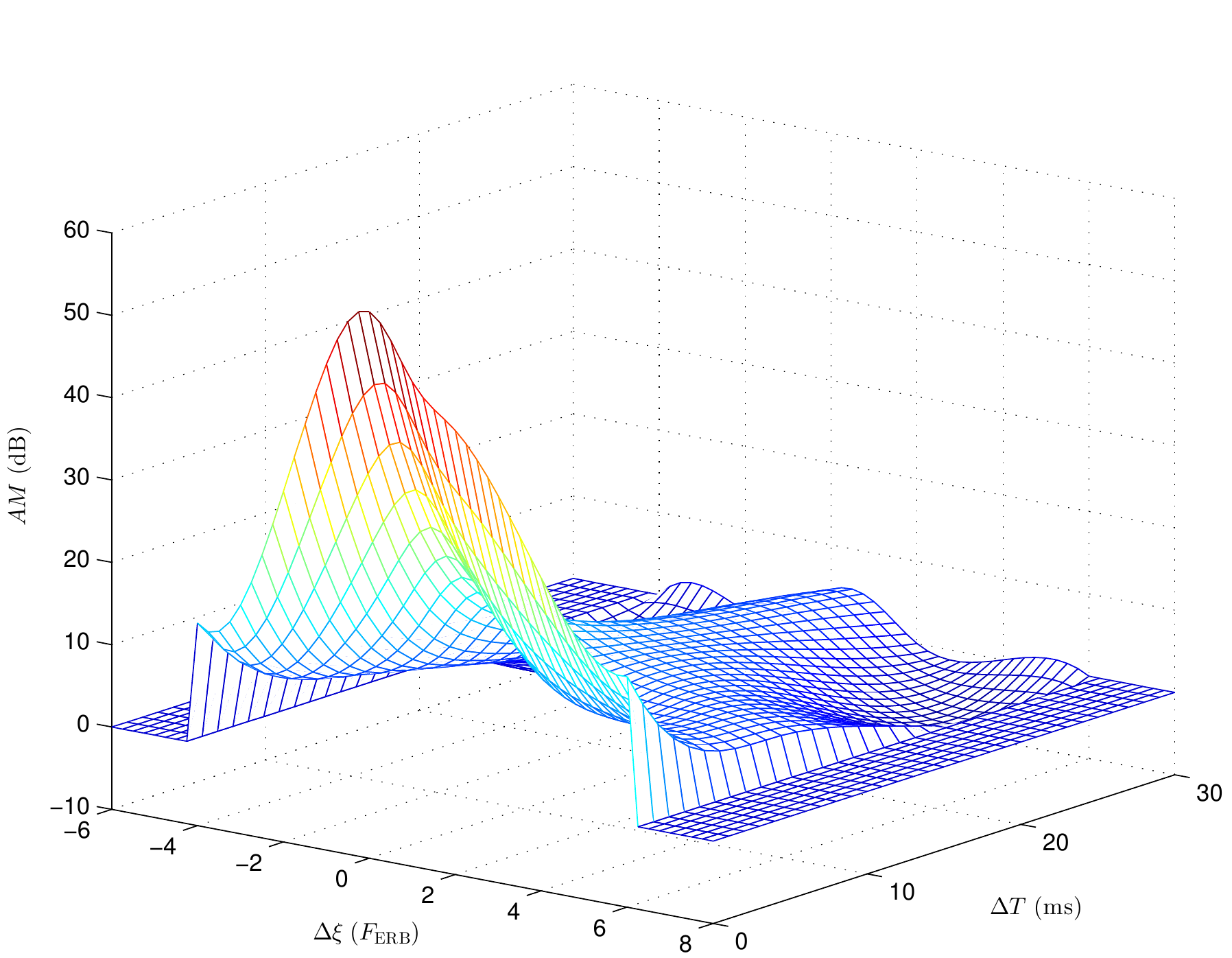}}
	  \end{center}
	  \caption{(a)~Conditions measured in \cite{Necciari:2010a} illustrated in the time-$F_{\mathrm{ERB}}$ plane. The gray circle symbolizes the masker atom $s_M(t)$. The blue circles symbolize the target atoms $s_T(t+\Delta T)$. The values of \dx$\,$ were -4, -2, -1, 0, +1, +2, +4, and +6~$F_{\mathrm{ERB}}$. The values of \dt$\,$ were 0, 5, 10, 20, and 30~ms. (b)~Mean data interpolated based on a cubic spline fit along the time-frequency plane. The \dt$\,$ axis was sampled at a step of 1~ms and the \dx$\,$ axis at a step of 0.25~$F_{\mathrm{ERB}}$. For \dx$\,$ coordinates outside the range of measurements a value of $AM$ = 0 was used.}
\end{figure}

\subsection{Computational auditory scene analysis} \label{sec:casa0}

The term auditory scene analysis (ASA), introduced by Bregman \cite{Bregman:1990a}, refers to the perceptual organization of auditory events into auditory streams. It is assumed that this perceptual organization constitutes the basis for the remarkable ability of the auditory system to separate sound sources, especially in noisy environments. A demonstration of this ability is the so-called ``cocktail party effect'', i.e. when one is able to concentrate on and follow a single speaker in a highly competing background (e.g. many concurring speakers combined with cutlery and glass sounds). The term computational auditory scene analysis (CASA) thus refers to the study of ASA by computational means \cite{wanbro06}. The CASA problem is closely related to the problem of source separation. Generally speaking, CASA systems can be considered as perceptually motivated sound source separators. The basic work flow of a CASA system is to first compute an auditory-based time-frequency transform (most systems use a 
gammatone filter bank, but any auditory representation that allows reconstruction can be used, see Sec.~\ref{sec:audlet0}). Second, some acoustic features like periodicity, pitch, amplitude and frequency modulations are extracted so as to build the perceptive organization (i.e. constitute the streams). Then, stream separation is achieved using so-called ``time-frequency masks''. These masks are directly applied to the perceptual representation; they retain the ``target'' regions (mask = 1) and suppress the background (mask = 0). Those masks can be binary or real, see e.g. \cite{wanbro06,Zhao:2012a}. The target regions are then re-synthesized by applying the inverse transform to obtain the signal of interest. Noteworthy, a perfect reconstruction transform is of importance here. Furthermore, the linearity and stability of the transform allow a separation of the audio streams directly in the transform domain. Most gammatone filter banks implemented in CASA systems are only approximately invertible, though. This 
is due to the fact that such systems implement 
gammatone filters in the analysis stage and their time-reversed impulse responses in the synthesis stage. This setting implies that the frequency response of the gammatone filter bank has an all-pass characteristic and features no ripple (equivalently in the frame context, that the system is tight, see~\ref{ssec:frametheory}). In practice, however, gammatone filter banks usually consider only a limited range of frequencies (typically in the interval 0.1--4~kHz for speech processing) and the frequency response features ripples if the filters' density is not high enough. If a high density of filters is used, the audio quality of the reconstruction is rather good~\cite{Strahl:2009a,Zhao:2012a}. Still, the quality could be perfect by using frame theory \cite{nebahopr15}. For instance, one could render the gammatone system tight (see Proposition~\ref{prop:cantight}) or use its dual frame (see Sec.~\ref{sec:perfrecfram0}).

The use of binary masks in CASA is directly motivated by the phenomenon of auditory masking explained above. However, time-frequency masking is hardly considered in CASA systems. As a final remark, an analogy can be established between the (binary) masks used in CASA and the concept of frame multipliers defined in Sec.~\ref{fmult}. Specifically, the masks used in CASA systems correspond to the symbol $m$ in \eqref{mult}. This analogy is not considered in most CASA studies, though, and offers the possibility for some future research connecting acoustics and frame multipliers. 

\section{\label{sec:frameth}Frame theory}

What is an appropriate setting for the mathematical background of audio signal processing? 
Since real-world signals are usually considered to have finite energy and technically are represented as functions of some variable (e.g. time), it is natural to think about them as elements of the space $L^2(\RR)$. Roughly speaking, $L^2(\RR)$ contains all functions $x(t)$ with finite energy, 
i.e. with $\|x\|^2= \int_{-\infty}^{+\infty} |x(t)|^2{\rm dt} < \infty$. For working with sampled signals, the analogue appropriate space is $\ell^2(K)$ ($K$ denoting a countable index set) 
which consists of the sequences $c = (c_k)_{k\in K}$ with finite energy,
 i.e. $\|c\|^2 = \sum_{k\in K} |c_k|^2<\infty$. 
 
Both spaces $L^2(\RR)$ and $\ell^2(K)$ are Hilbert spaces  and one may use the rich theory 
ensured by the availability of an inner product, that serves as a measure of correlation, and is used to define orthogonality, of elements in the Hilbert space.
In particular, the inner product enables the representation of all functions in $\mathcal H$ in terms of their inner products with a set of reference functions:
A standard approach for such representations uses orthonormal bases (ONBs), see e.g. \cite{heilbook}.
Every separable Hilbert space $\Hil$ has an ONB $(e_k)_{k\in K}$ and every element $x\in\Hil$ can be written as 
\begin{equation}\label{repronb}
x=\sum_{k\in K}\<x,e_k\>e_k
\end{equation}
 with uniqueness of the coefficients $\<x,e_k\>$, $k\in K$. 
The convenience of this approach is that there is a clear (and efficient) way for calculating the coefficients in the representations using the same orthonormal sequence. 
Even more, the energy in the coefficient domain (i.e.,  \emph{ the square of the $\ell^2$-norm}) is exactly the energy of the element $x$: 
\begin{equation}
\sum_{k\in K} |\<x,e_k\>|^2=\|x\|^2. \tag{Parseval equality}
\end{equation} 
Furthermore,  the representation (\ref{repronb}) is stable - 
if the coefficients  $(\<x,e_k\>)_{k\in K}$ are slightly changed to $(a_k)_{k\in K}\in\ell^2$, one obtains an element $\widetilde{x}=\sum_{k\in K} a_k e_k$ close to the original one $x$. 

However, the use of ONBs has several disadvantages. Often the construction of orthonormal bases with some given side constraints is difficult or even impossible (see below). 
\lq\lq Small perturbation\rq\rq 
of the orthonormal basis' elements may destroy the orthonormal structure \cite{Young}. 
Finally, the uniqueness of the coefficients in (\ref{repronb}) leads to a lack of exact 
 reconstruction when some of these coefficients are lost or disturbed during transmission. \\

This naturally leads to the question how the concept of ONBs could be generalized to overcome those disadvantages. 
As an extension of the above-mentioned Parseval equality for ONBs, one could consider inequalities instead of an equality, i.e. boundedness from above and below  (see Def. \ref{framedef}). This leads to the concept of \emph{frames}, which was introduced by Duffin and Schaeffer \cite{dusc52} in 1952. It took several decades for scientists to realize the importance and applicability of frames. 
Popularized around the 90s in the wake of wavelet theory \cite{dagrme86,hewa89,da92}, frames have seen increasing interest and extensive investigation by many researchers ever since.
Frame theory is both a beautiful abstract mathematical theory and a concept applicable in many other disciplines like e.g. engineering, medicine, and psychoacoustics, see Sec.~\ref{sec:appli}. 

Via frames, one can avoid the restrictions of ONBs while keeping their important properties.
Frames still allow perfect and stable reconstruction of all the elements of the space, though the representation-formulas in general are not as simple as the ones via an ONB 
(see Sec.~\ref{sec:perfrecfram0}). Compared to orthonormal bases, the frame property itself is much more stable under perturbations (see, e.g.,  \cite[Sec. 15]{ole1}).
Also, in contrast to orthonormal bases, frames allow redundancy which  
is desirable e.g. in signal transmission, for reconstructing signals when some coefficients are lost, and for noise reduction. 
Via redundant frames one has multiple representations and this allows to choose appropriate coefficients fulfilling particular constraints, e.g. when aiming at sparse representations. 
Furthermore, frames can be easier and faster to construct than ONBs. 
Some advantageous side constraints can \emph{only} be fulfilled for frames. For example, Gabor frames provide convenient and efficient signal processing tools, but good localization in both time and frequency can never be achieved if the Gabor frame is an ONB or even a Riesz basis (cf. Balian-Low Theorem, see e.g. \cite[Theor. 4.1.1]{ole1}), while redundant Gabor frames for this purpose are easily constructed (for example using the Gaussian function). See Sec.~\ref{sec:tfmask0} on how good localization in time and frequency is important in masking experiments. 

Some of the main properties of frames were already obtained in the first paper \cite{dusc52}. For extensive presentation on frame theory, we refer to \cite{Casaz1,gr01,ole1,heilbook}.\\

In this section we collect 
the basics of frame theory relevant to the topic of the current paper. 
All the statements presented here are 
well known. 
 Proofs are given just to make the paper self-contained, for convenience of the readers, and to facilitate a better understanding of the mathematical concepts. 
They are mostly based on \cite{ole1,dusc52,gr01}. Throughout the rest of the section, $\Hil$ denotes a separable Hilbert space with inner product $\<\cdot,\cdot\>$, ${\rm Id}_\Hil$ - the identity operator on $\Hil$, 
$K$ - a countable index set, and $\Phi$ (resp. $\Psi$) - a sequence $(\phi_k)_{k\in K}$ (resp. $(\psi_k)_{k\in K}$) with elements from $\Hil$. The term \emph{operator} is used for a linear mapping. Readers not familiar with Hilbert space theory can simply assume $\mathcal H = \bd L^2(\RR)$ for the remainder of this section.

\subsection{Frames: A Mathematical viewpoint} 

The frame concept extends naturally the Parseval equality permitting inequalities, 
i.e., the ratio of the energy in the coefficient domain to the energy of the signal may be bounded from above and below instead of being necessarily one:

\begin{Def}\label{framedef} 
A countable sequence $\Phi = (\phi_k)_{k\in K}$ is called a \emph{frame} for the Hilbert space $\Hil$ if there exist positive constants $A$ and $B$  such that 
\begin{equation}\label{frdef}
 A \cdot \norm{\Hil}{x}^2 \le \sum \limits_{k\in K} \left| \left< x, \phi_k \right> \right|^2 \le B \cdot  \norm{\Hil}{x}^2 , \ \forall \ x \in \Hil.
\end{equation}
The constant $A$ (resp. $B$) is called a \emph{lower} (resp. \emph{upper}) \emph{frame bound} of $\Phi$. A frame is called \emph{tight} \emph{with frame bound $A$} if $A$ is both a lower and an upper frame bound. 
A tight frame with bound $1$ is called a \emph{Parseval frame}.
\end{Def}

Clearly, every ONB is a frame, but not vice-versa. 
Frames can naturally be split into two classes - the frames which 
 still fulfill a basis-property, and the ones that do not: 
\begin{Def}\label{defrb}
A frame $\Phi$ for $\Hil$ which is a Schauder basis\footnote{A sequence $\Phi$ is called a \emph{Schauder basis} for $\Hil$ if every element $x\in\Hil$ can be written as $x=\sum_{k\in K} c_k \phi_k$ with unique coefficients $(c_k)_{k\in K}$.} for $\Hil$ is called a \emph{Riesz basis} for $\Hil$. 
  A frame for $\Hil$ which is not a Schauder basis for $\Hil$ is called \emph{redundant} (also called \emph{overcomplete}). 
\end{Def} 

Note that Riesz bases were introduced by Bari \cite{Bari51} in different but equivalent ways. 
Riesz bases also extend ONBs, but contrary to frames, Riesz bases still have the disadvantages resulting from the basis-property, as they do not allow redundancy.
For more on Riesz bases, see e.g.  \cite{Young}. 
As an illustration of the concepts of ONBs, Riesz bases, and redundant frames in a simple setting, consider examples in the Euclidean plane, see Fig. \ref{onbrbfr}.

\begin{figure}
(a) ONB $(e_1, e_2)$ for $\RR^2$ \hspace{1.5cm} 
 (b)  unique representation of $x$ via $(e_1, e_2)$\\
\includegraphics[width=0.89\textwidth]{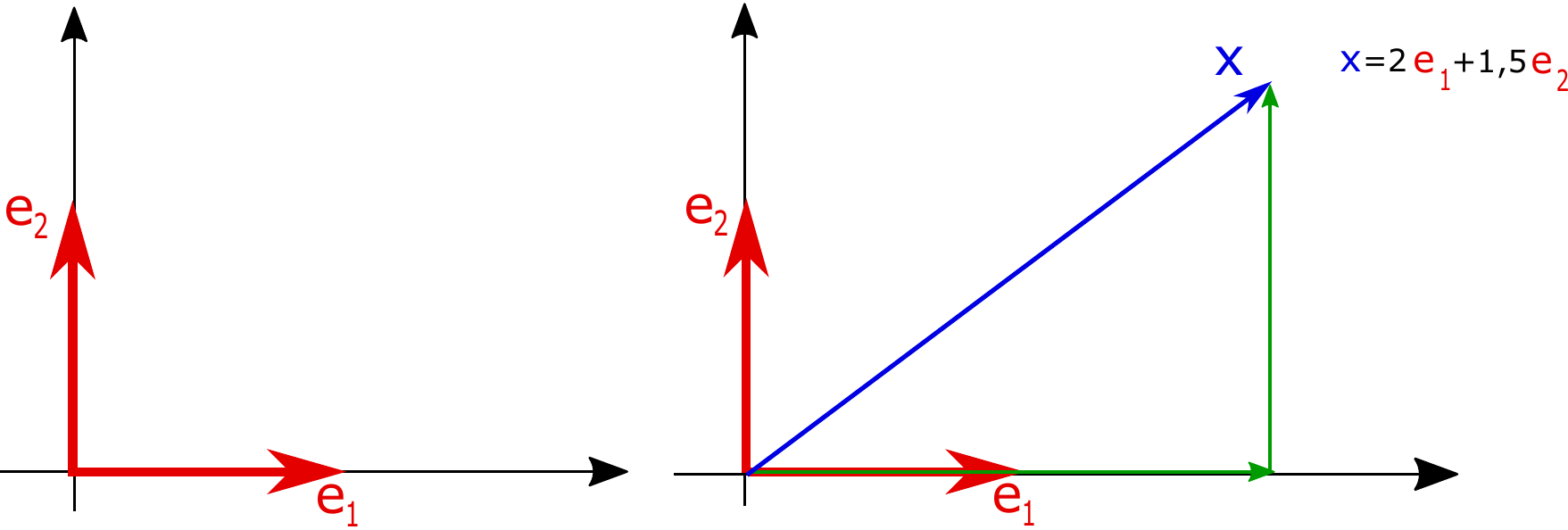} \\
(c) Riesz basis $(\phi_1,\phi_2)$ for $\RR^2$ \hspace{.8cm} (d) unique representation of $x$ via $(\phi_1,\phi_2)$\\
 \includegraphics[width=0.89\textwidth]{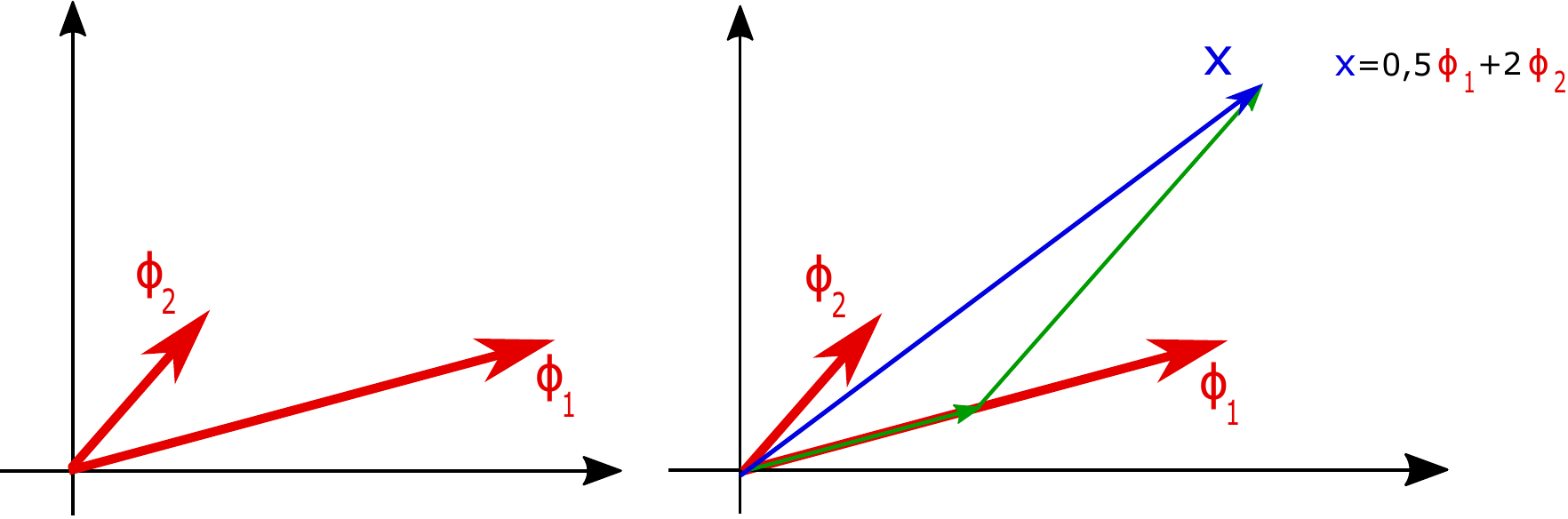}\\ 
 (e) frame $(\phi_1,\phi_2, \phi_3)$ for $\RR^2$ \hspace{1cm} (f) non-unique representation of $x$ via $(\phi_1,\phi_2, \phi_3)$\\
\includegraphics[width=1\textwidth]{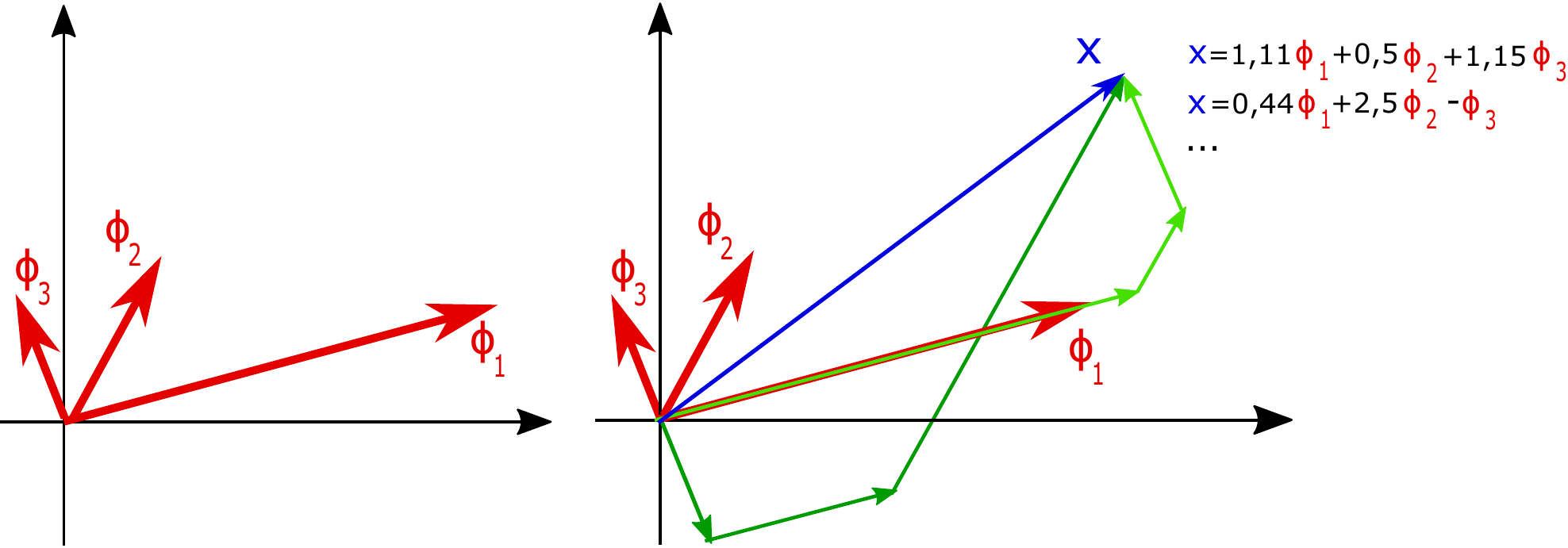} 
\caption{Examples in $\RR^2$: ONB (a,b), Riesz basis (c,d), frame (e,f)} \label{onbrbfr}
\end{figure}

Note that in a finite dimensional Hilbert space, considering only finite sequences, 
 frames are precisely the complete sequences  
(see, e.g., \cite[Sec. 1.1]{ole1}), i.e., the sequences which span the whole space. 
However, this is not the case in infinite-dimensional Hilbert spaces - every frame is complete, 
 but completeness is not sufficient to establish the frame property \cite{dusc52}. For results focused on  frames in finite dimensional spaces, refer to \cite{ba08-1,casazza2012finite}.

As non-trivial examples, let us mention a specific type of frames used often in signal processing applications, namely Gabor frames.
A Gabor system is comprised of atoms of the form
 \begin{equation} \label{Gabsystem}
  g_{\omega,\tau}(t)= e^{2\pi {\rm i} \omega t} g(t - \tau),
\end{equation}
with function $g\in L^2(\RR)$ called the \emph{(generating)}  \emph{window} and with time- and frequency-shift $\tau,\omega\in\RR$, respectively. 
To allow perfect and stable reconstruction, the Gabor system  $( g_{\omega,\tau} )_{\omega,\tau\in K(\subset \RR^2)}$ is assumed to have the frame-property and in this case is called a \emph{Gabor frame}. 
Note that the analysis operator of a Gabor frame corresponds to a \emph{sampled Short-Time-Fourier transform} (see, e.g., \cite{gr01}) also referred to as \emph{Gabor transform}.

Most commonly, \emph{regular Gabor frames} are used; these are frames of the form
$ ( g_{k,l} )_{k,l\in \ZZ} = \left(e^{2\pi {\rm i} kb\cdot} g(\cdot - la)\right)_{k,l\in \ZZ}$ for some positive $a$ and $b$ satisfying necessarily (but in general not sufficiently ) $ab\leq 1$. 
To mention a concrete example - for the Gaussian $g(t)=e^{-t^2}$, the respective regular Gabor system 
 $ ( g_{k,l} )_{k,l\in \ZZ}$ is a frame for $L^2(\RR)$ 
if and only if $ab<1$ (see, e.g., \cite[Sec. 7.5]{gr01} and references therein).

Other possibilities include using alternative sampling structures, on subgroups \cite{hosowi13} or irregular sets \cite{cach03}. If the window is allowed to change with time (or frequency) one obtains the non-stationary Gabor transform \cite{badohojave11}. There it becomes apparent that frames allow to create adaptive and adapted transforms \cite{badokoto13}, while still guaranteeing perfect reconstruction.

If not continuous  but sampled signals are considered, Gabor theory works similarly. 
\emph{Discrete Gabor frames} can be defined in an analogue way, namely, frames of the form $\left(e^{2\pi {\rm i}  k/M\cdot} h[\cdot - la]\right)_{l\in\ZZ, k=0,1,\ldots, M-1}$ for $h\in \ell^2(\ZZ)$ with $a,M\in\NN$, where $a/M\leq 1$ is necessary for the frame property. 
 For readers interested in the theory of Gabor frames on $\ell^2(\ZZ)$, see, e.g., \cite{vk}. For constructions of discrete Gabor frames from Gabor frames for $L^2(\RR)$ through sampling, refer to \cite{Jan97,sond05}.

\subsubsection{Frame-related operators} \label{froperators}
Given a frame $\Phi$ for $\Hil$, consider the following linear mappings:
\begin{eqnarray} 
 \mbox{\emph{Analysis operator}:\hspace{.1cm}} &  \bd C_\Phi : \Hil \rightarrow l^2(K), & \mbox{$\bd C_\Phi x := (\left< x, \phi_k \right>)_{k\in K}$}; \hfill \nonumber \\
  \mbox{\emph{Synthesis operator}:} &  \bd D_\Phi : l^2(K) \rightarrow \Hil, &  \mbox{$\bd D_\Phi (c_k)_{k\in K} := \sum_{k\in K} c_k \phi_k$}; \hfill \nonumber\\
  \mbox{\emph{Frame operator}:\hspace{.43cm}} &  \bd S_\Phi : \Hil \rightarrow \Hil, \hspace{.39cm} &  \mbox{$\bd S_\Phi x:= \bd D_\Phi\bd C_\Phi x= \sum_{k\in K}  \left<x , \phi_k \right> \phi_k$}. \hfill \label{eq:frameop} 
\end{eqnarray} 
These operators are tremendously important for the theoretical investigation of frames as well as for signal processing.
As one can observe, the analysis (resp. synthesis, frame) operator corresponds to analyzing (resp. synthesizing, analyzing and re-synthesizing) a signal. 
In the following statement the main properties of the frame-related operators are listed.

\begin{Thm}\label{froper} {\rm (e.g. \cite[Sec. 5]{ole1})} Let $\Phi$ be a frame for $\Hil$ with frame bounds $A$ and $B$ ($A\leq B$). 
Then the  following holds.
\begin{itemize}
\item[{\rm (a)}]  $\bd C_\Phi$ is a bounded injective 
operator with bound $\|\bd C_\Phi\|\leq \sqrt{B}$.
\item[{\rm (b)}]  $\bd D_\Phi$ is a bounded surjective
 operator with bound $\|\bd D_\Phi\|\leq \sqrt{B}$ and 
$\bd D_\Phi=\bd C_\Phi^*$.
\item[{\rm (c)}]   $\bd S_\Phi$ is a bounded bijective positive 
 self-adjoint operator with $\|\bd S_\Phi\|\leq B$. 
\item[{\rm (d)}]  $(\bd S_\Phi^{-1} \phi_k)_{k\in K}$ is a frame for $\Hil$ with frame bounds  $1/B, 1/A$.  
\end{itemize}
\end{Thm}
\begin{proof} 
(a) By the frame inequalities (\ref{frdef}) we have 
$\sqrt{A}\|x\|_\Hil \leq \|\bd C_\Phi x\|_{\ell^2}\leq \sqrt{B} \|x\|_{\Hil}$ for every $x\in\Hil$; the upper inequality implies the boundedness and the lower one - the injectivity, i.e. the operator is one-to-one. 

(b)  First show that  $\bd D_\Phi$ is well defined, i.e.,
 that $\sum_{k\in K} c_k \phi_k$ converges for every $(c_k)_{k\in K}\in\ell^2(K)$. Without loss of generality, for simplicity of the writing, we may denote $K$ as $\NN$. 
Fix arbitrary $(c_k)_{k\in\NN}\in\ell^2$. For every $p,q\in\NN$, $p>q$, 
\begin{eqnarray*}
 \|\sum_{k=1}^p c_k \phi_k - \sum_{k=1}^q c_k \phi_k\|_{\Hil}
&=& \sup_{x\in\Hil, \|x\|_{\Hil}=1} | \< \sum_{k=q+1}^p c_k \phi_k, x\> | \\
&\leq & \sup_{x\in\Hil, \|x\|_{\Hil}=1} (\sum_{k=q+1}^p |c_k|^2)^{1/2} (\sum_{k=q+1}^p |\<\phi_k, x\>|^2)^{1/2} \\
&\leq& \sqrt{B} \, (\sum_{k=q+1}^p |c_k|^2)^{1/2}\xrightarrow[p,q\to \infty]{}0,
\end{eqnarray*}
which implies that $\sum_{k=1}^p c_k \phi_k$ converges in $\Hil$ as $p\to\infty$.
Using the adjoint of $\bd C_\Phi$, for every $(c_k)_{k=1}^\infty\in\ell^2$ and every $y\in\Hil$, one has that
$$ \<\bd C^*_\Phi (c_k)_{k=1}^\infty, y\> = \<\bd (c_k)_{k=1}^\infty, \bd C_\Phi y\> 
= \sum_{k=1}^\infty c_k \overline{\<y,\phi_k\> } = \sum_{k=1}^\infty c_k \< \phi_k, y\> = \< \sum_{k=1}^\infty c_k \phi_k, y\>.$$
Therefore $\bd D_\Phi = \bd C_\Phi^*$, implying also the boundedness of $D_\Phi$. 

For every $x\in\Hil$, we have $\|\bd D_\Phi^* x\|_{\ell^2}=\|\bd C_\Phi x\|_{\ell^2} \geq \sqrt{A}\|x\|$, 
which implies (see, e.g., \cite[Theorem 4.15]{Rudin}) that $\bd D_\Phi$ is surjective, i.e. it maps onto the whole space $\Hil$. 

(c) The boundedness and self-adjointness of $\bd S_\Phi$ follow from (a) and (b). Since,
   $\<\bd S_\Phi x,x\> 
=  \sum_{k\in K} |\<x,\phi_k\>|^2$,  $\bd S_\Phi$ is positive and the frame inequalities (\ref{frdef}) 
mean that
\begin{equation} \label{sineq}
A\|x\|_\Hil^2 \leq \<\bd S_\Phi x,x\>  \leq B \|x\|_\Hil^2, \forall x\in\Hil,
\end{equation}
implying that
$0\leq \<({\rm Id}_\Hil - \frac{1}{B} \bd S_\Phi)x,x\> \leq \frac{B-A}{B} \|x\|^2_\Hil$ for all $x\in\Hil$. Then the norm of the bounded self-adjoint operator ${\rm Id}_\Hil - \frac{1}{B} \bd S_\Phi$ satisfies 
 $$ \|{\rm Id}_\Hil - \frac{1}{B} \bd S_\Phi\|
=
\sup_{x\in\Hil, \|x\|_\Hil=1} \<({\rm Id}_\Hil - \frac{1}{B} \bd S_\Phi)x,x\>
\leq
 \frac{B-A}{B} <1,$$
which by the Neumann theorem (see, e.g., \cite[Theor. 8.1]{Heuser})
implies that $\bd S_\Phi$ is bijective.

(d) As a consequence of (c), $S_\Phi^{-1}$ is bounded, self-adjoint, and positive.
In the language of partial ordering of self-adjoint operators (see, e.g., \cite[Sec. 68]{Heuser}), (\ref{sineq}) can be written as 
\begin{equation} \label{sineq2}
A \cdot {\rm Id}_\Hil  \leq \bd S_\Phi  \leq B\cdot {\rm Id}_\Hil. 
\end{equation}
Since $\bd S_\Phi^{-1}$ is positive and commutes with $\bd S_\Phi$ and ${\rm Id}_\Hil$, one can multiply the inequalities in (\ref{sineq2}) with $\bd S_\Phi^{-1}$ (see, e.g., \cite[Prop. 68.9]{Heuser})
and obtain 
\begin{equation*} \label{sineq3}
\frac{1}{B} {\rm Id}_\Hil \leq \bd S_\Phi^{-1}  \leq \frac{1}{A} {\rm Id}_\Hil,
\end{equation*}
which means that 
\begin{equation} \label{sineq4}
\frac{1}{B} \|x\|_\Hil^2 \leq \<\bd S_\Phi^{-1}x, x\>  \leq \frac{1}{A} \|x\|^2_\Hil, \ \forall x\in\Hil. 
\end{equation}
For every $x\in\Hil$, denote $y_x=\bd S_\Phi^{-1} x$ and use the fact that $\bd S_\Phi^{-1}$ is self-adjoint to obtain
$$\sum_{k\in K} |\<x, \bd S_\Phi^{-1} \phi_k\>|^2 
= \sum_{k\in K} |\<  y_x,  \phi_k\>|^2 
= \<y_x, \bd S_\Phi y_x\>= \<\bd S_\Phi^{-1} x,x\>.
$$
Now (\ref{sineq4}) completes the conclusion that $(\bd S_\Phi^{-1} \phi_k)_{k\in K}$ is a frame for $\Hil$ with frame bounds $1/B$, $1/A$.
\end{proof}

\subsubsection{Perfect reconstruction via frames} \label{sec:perfrecfram0}

Here we consider one of the most important properties of frames, namely, the possibility to have perfect reconstruction of all the elements in the space.

\begin{Thm}\label{frexp} {\rm (e.g. \cite[Corol. 5.1.3]{gr01})}
Let $\Phi$ be a frame for $\Hil$. Then there exists a frame $\Psi$ for $\Hil$ such that
\begin{equation}\label{frrepr}
x=\sum_{k\in K}\<x, \psi_k\>\phi_k = \sum_{k\in K}\<x, \phi_k\>\psi_k, \ \forall x\in\Hil.
\end{equation}
\end{Thm}
\begin{proof} 
By Theorem \ref{froper}(d), the sequence  $(\bd S_\Phi^{-1} \phi_k)_{k\in K}$ is a frame for $\Hil$. Take $\Psi:=(\bd S_\Phi^{-1} \phi_k)_{k\in K}$. 
Using the boundedness and the self-adjointness of $S_\Phi$, for every $x\in \Hil$, 
$$ \sum_{k\in K}\<x, \phi_k\>\psi_k = \sum_{k\in K}\<x, \phi_k\>S_\Phi^{-1} \phi_k = 
S_\Phi^{-1}\sum_{k\in K}\<x, \phi_k\>\phi_k =
 S_\Phi^{-1} S_\Phi x=x,
$$
$$ \sum_{k\in K}\<x, \psi_k\>\phi_k = \sum_{k\in K}\<x, S_\Phi^{-1}\phi_k\>\phi_k
= \sum_{k\in K}\<S_\Phi^{-1}x, \phi_k\>\phi_k= S_\Phi S_\Phi^{-1}x=x. 
$$
\end{proof}

Let $\Phi$ be a frame for $\Hil$. 
Any frame $\Psi$ for $\Hil$, which satisfies (\ref{frrepr}), is called a \emph{dual frame} of $\Phi$. By the above theorem, every frame has at least one dual frame, namely, the sequence
\begin{equation}\label{eq:candual}
(S_\Phi^{-1} \phi_k)_{k\in K},
\end{equation} called the \emph{canonical dual of $\Phi$}. 
When the frame is a Riesz basis, then the coefficient representation is unique and thus there is only one dual frame, the canonical dual. 
When the frame is redundant, then there are other dual frames different from the canonical dual (see, e.g., \cite[Lemma 5.6.1]{ole1}), even infinitely many. This provides multiple choices for the coefficients in the frame representations, which is desirable in some applications (see, e.g., \cite{badokoto13}).
The canonical dual has a minimizing property in the sense that the coefficients 
$(\<x, S_\Phi^{-1} \phi_k\>)_{k\in K}$ in the representation
$x=\sum_{k\in K} \<x, S_\Phi^{-1} \phi_k\> \phi_k$
have the minimal $\ell^2$-norm compared to the coefficients $(c_k)_{k\in K}$ in all other possible representations
$x=\sum_{k\in K} c_k \phi_k$. However, for certain applications other constraints are of interest - e.g. sparsity, efficient algorithms for representations or particular shape restrictions on the dual window \cite{elsuwe05,persampta13}. 
The canonical dual is not always efficient to calculate nor does it always have the desired structure; in such cases other dual frames are of interest \cite{chr06,hahule11,bole07}. The particular case of tight frames is very convenient for efficient reconstructions, because the canonical dual is simple and does not require operator-inversion:

\begin{Cor}\label{dualtight} {\rm (e.g. \cite[Sec. 5.7]{ole1})}
The canonical dual of a tight frame $(\phi_k)_{k\in K}$ with frame bound $A$ is the sequence $(\frac{1}{A}\phi_k)_{k\in K}$.
\end{Cor}
\begin{proof}
Let $\Phi$ be a tight frame for $\Hil$ with frame bound $A$. It follows from (\ref{sineq2}) that $\bd S_\Phi=A \cdot{\rm Id}_\Hil$ and thus the canonical dual of $\Phi$ is $(\bd S_\Phi^{-1}\phi_k)_{k\in K}=(\frac{1}{A}\phi_k)_{k\in K}$.
\end{proof}

In acoustic applications, it can be of big advantage to not be forced to distinguish between analysis and synthesis atoms.
So, one may aim to do analysis and synthesis with the same sequence as an analogue to the case with ONBs. 
However, such an analysis-synthesis strategy would perfectly reconstruct all the elements of the space if and only if this sequence is a Parseval frame: 

\begin{Pro} \label{sec:parseval1} {\rm (e.g. \cite[Lemma 5.7.1]{ole1})} The sequence $\Phi$ satisfies
\begin{equation}\label{parseq}
x=\sum_{k\in K}\<x,\phi_k\>\phi_k, \ \forall x\in\Hil,
\end{equation}
if and only it is a Parseval frame for $\Hil$.
\end{Pro}
\begin{proof} 
Let $\Phi$ be a Parseval frame for $\Hil$.  By Corollary \ref{dualtight}, the canonical dual of $\Phi$ is the same sequence $\Phi$, which implies that (\ref{parseq}) holds. 
Now assume that (\ref{parseq}) holds. Then for every $x\in\Hil$,
$$\|x\|^2=\<\sum_{k\in K}\<x,\phi_k\>\phi_k, x\>=\sum_{k\in K}\<x,\phi_k\>\<\phi_k, x\>
=\sum_{k\in K}|\<x,\phi_k\>|^2,$$
which means that $\Phi$ is a Parseval frame for $\Hil$.
\end{proof}

The above statement characterizes the sequences which provide reconstructions exactly like ONBs - these are precisely the Parseval frames. A trivial example of such a frame which is not an ONB is the sequence \sloppy 
$(e_1, e_2/\sqrt{2},  e_2/\sqrt{2},  e_3/\sqrt{3},  e_3/\sqrt{3},  e_3/\sqrt{3},  \ldots)$, where $(e_k)_{k=1}^\infty$ denotes an ONB for $\Hil$. Clearly, any tight frame with frame bound $A$ is easily converted into a Parseval frame by dividing the frame elements by the square root of $A$. Given any frame, one can always construct a Parseval frame as follows:

\begin{Pro}\label{prop:cantight} {\rm (e.g. \cite[Theor. 5.3.4]{ole1})}
Let $\Phi$ be a frame for $\Hil$. Then  $\bd S_\Phi^{-1}$ has a positive square root 
 and $(\bd S_\Phi^{-1/2} \phi_k)_{k\in K}$ forms a Parseval frame for $\Hil$.
\end{Pro}
\begin{proof}  
Since $\SSphi^{-1}$ is a bounded positive self-adjoint operator, there is a unique bounded positive self-adjoint operator, 
which is denoted by  $\SSphi^{-1/2}$, with $\SSphi^{-1} = \SSphi^{-1/2} \SSphi^{-1/2}$. Furthermore,   $\SSphi^{-1/2}$ commutes with $\SSphi$.
For every $x \in \Hil$,
$$ \sum \limits_{k\in K} \< x, \SSphi^{-1/2} \phi_k \> \SSphi^{-1/2}\phi_k = \SSphi^{-1/2} \sum \limits_{k\in K} \< \SSphi^{-1/2} x, \phi_k \> \phi_k = \SSphi^{-1/2} \SSphi \SSphi^{-1/2} x = 
\SSphi^{-1} \SSphi  x =
x.  $$ 
By Proposition \ref{sec:parseval1} this means that $(\bd S_\Phi^{-1/2} \phi_k)_{k\in K}$ is a Parseval frame for $\Hil$.
\end{proof}

Finally, note that frames guarantee stability. 
Let $\Phi$ be a frame for $\Hil$ with frame bounds $A,B$. Then 
$\sqrt{A} \|x-y\|\leq  \|(\<x,\phi_k\>)-(\<y,\phi_k\>)_{k\in K} \|_{\ell^2}\leq \sqrt{B} \|x-y\|$ for $x,y\in\Hil$, which implies that 
close signals lead to close analysis coefficients and vice versa. Furthermore, the representations via $\Phi$ and a dual frame $\Psi$ is stable. If a signal $x$ is transmitted via the coefficients $(\<x,\psi_k\>)_{k\in K}$ but, during transmission, the coefficients are slightly disturbed (i.e. modified to a sequence $(a_k)_{k\in K}\in\ell^2$ with small $\ell^2$-difference), then by Theorem \ref{froper}(b) the \lq\lq reconstructed\rq\rq \, signal $y=\sum_{k\in K} a_k \phi_k$ will be close to $x$: $\|x-y\|=\|\sum_{k\in K} (\<x,\psi_k\>-a_k)_{k\in K}\phi_k\| \leq \sqrt{B} \|(\<x,\psi_k\>-a_k)_{k\in K} \|_{\ell^2}$. 

\subsection{Frame multipliers} \label{fmult}

Multipliers have been used implicitly for quite some time in applications, as time-variant filters, see e.g. \cite{hlma98}.
The first systematic theoretical development of Gabor multipliers appeared in \cite{feno03}. An extension of the multiplier concept to general frames in Hilbert spaces was done in \cite{ba07} and it can be derived as an easy consequence of Theorem \ref{froper}:

\begin{Pro} {\rm \cite{ba07}}
Let $\Phi$  and $\Psi$ be frames for $\Hil$ and let $m=(m_k)_{k\in K}$ be a complex scalar sequence in $\ell^\infty(K)$. Then the series
$\sum_{k\in K} m_k \<x, \psi_k\> \phi_k$
converges 
for every $x\in\Hil$ and determines a bounded operator on $\Hil$. 
\end{Pro}
\begin{proof} 
For every $x\in\Hil$, Theorem \ref{froper}(a) implies that $(\<x, \psi_k\>)_{k\in K}\in\ell^2$ and thus $(m_k\<x, \psi_k\>)_{k\in K}\in\ell^2$, which by  Theorem \ref{froper}(b)  implies that the series $\sum_{k\in K} m_k \<x, \psi_k\> \phi_k$ converges.  
Thus, the mapping 
 $\bd M_{m,\Phi,\Psi}$ determined by $\bd M_{m,\Phi,\Psi}x:=\sum_{k\in K} m_k \<x, \psi_k\> \phi_k$ is well defined on $\Hil$ and furthermore linear. For every $x\in \Hil$, 
\begin{eqnarray*}
\|\bd M_{m,\Phi,\Psi}x\|_{\Hil} &= &\|\bd D_\Phi (m_k \<x, \psi_k\>)_{k\in K}\|_\Hil 
\leq \|\bd D_\Phi\| \cdot \|(m_k \<x, \psi_k\>)_{k\in K}\|_{\ell^2} \\
& \leq & \|\bd D_\Phi\| \cdot\|m\|_\infty \cdot\| \bd C_\Psi \|\cdot \|x\|_\Hil,
\end{eqnarray*}
implying the boundedness of $\bd M_{m,\Phi,\Psi}$. 
\end{proof}

Due to above proposition, frame multipliers can be defined as follows:

\begin{Def} \label{frmult}
Given frames $\Phi$  and $\Psi$  for $\Hil$ and given complex scalar sequence $m=(m_k)_{k\in K}\in\ell^\infty(K)$, the operator $M_{m,\Phi,\Psi}$ determined by
\begin{equation}\label{mult}
\bd M_{m,\Phi,\Psi}x:=\sum_{k\in K} m_k \<x, \psi_k\> \phi_k, \ x\in\Hil,
\end{equation}
is called a \emph{frame multiplier} with a \emph{symbol} $m$.
\end{Def}

Thus, frame multipliers extend the frame operator, allowing different frames for the analysis and synthesis step, and modification in between (for an illustration, see Figure \ref{fig:exampmult1}). However, in contrast to frame operators, multipliers in general loose the bijectivity (as well as self-adjointness and positivity). 
For some applications it might be necessary to invert multipliers, which brings the interest to bijective multipliers and formulas for their inverses  - for interested readers, we refer to 
\cite{bsreprinv2015,bast12,bast13,iwota11} for some investigation in this direction.

\begin{figure}[h!t] 
\includegraphics[width=1\textwidth]{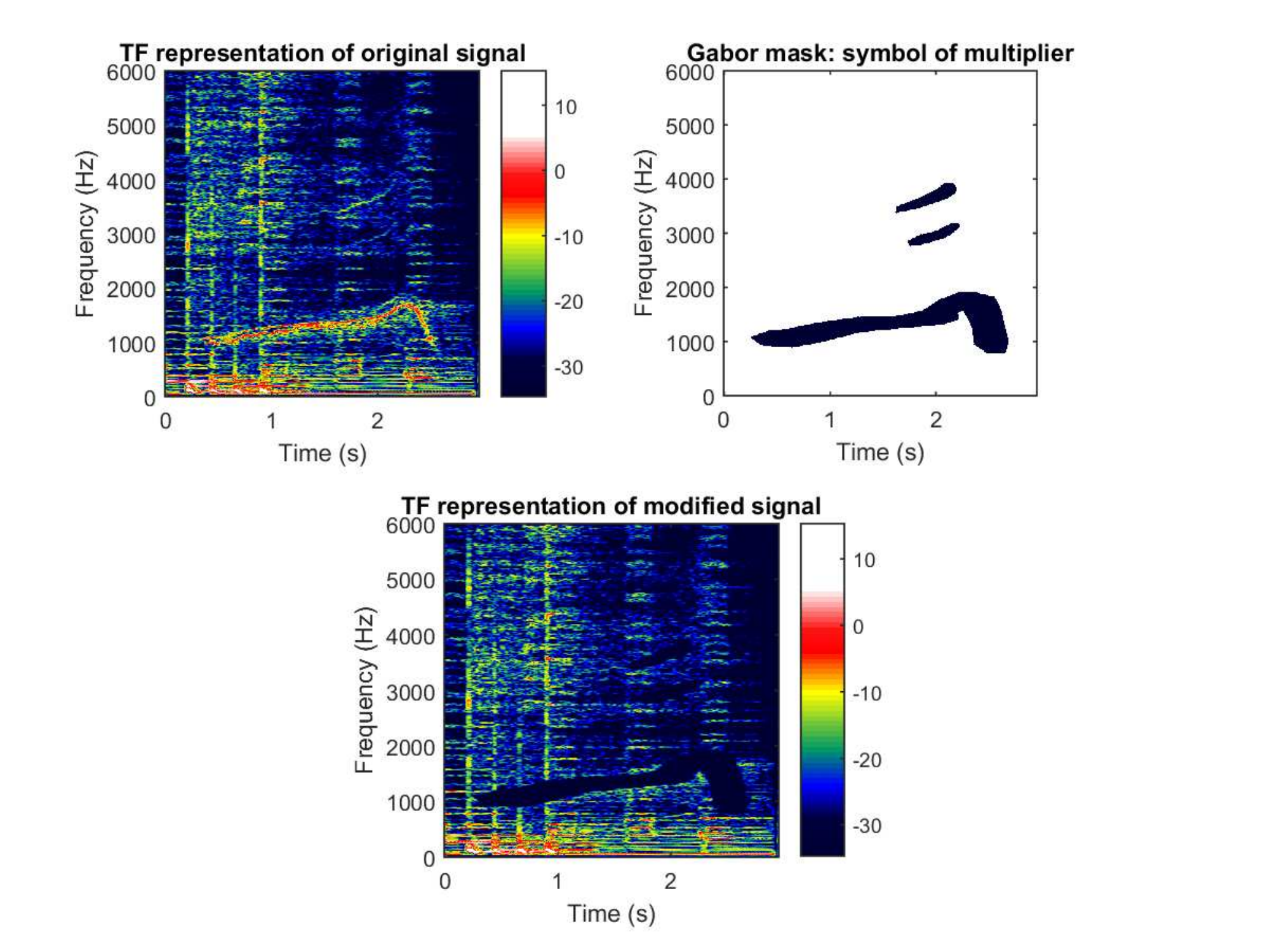} 
\caption{\label{fig:exampmult1} {\em An illustrative example to visualize a multiplier (taken from \cite{bsreprinv2015}).} 
(TOP LEFT) The time-frequency representation of 
 the music signal $f$. (TOP RIGHT) The symbol $m$, found by a (manual) estimation of the time-frequency region of the singer's voice. 
 (BOTTOM)  Time-frequency representation of $M_{m,\widetilde \Psi,\Psi}f$.} 
\end{figure}

In the language of signal processing,  Gabor filters \cite{hlawatgabfilt1} are a particular way to do time-variant filtering. In fact, Gabor filters are nothing but frame multipliers associated to a Gabor frame.
A signal $x$ is transformed to the time-frequency domain (with a Gabor frame $\Phi$), then modified there by point-wise multiplication with the symbol $m$, followed by re-synthesis via some Gabor frame $\Psi$ providing a modified signal. If some elements $m_k$ of the symbol $m$ are zero, the corresponding coefficients are removed, as sometimes used in applications like CASA or percerptual sparsity, see Secs.~\ref{sec:casa0} and~~\ref{sec:Irrel0}. 

\subsubsection{Implementation}
In the finite-dimensional case, frames lend themselves easily to implementation in computer codes \cite{ba08-1}. 
The Large Time-Frequency Analysis Toolbox (LTFAT) \cite{basoto12}, see \url{http://ltfat.github.io/}, is an open-source Matlab/Octave toolbox intended for time-frequency analysis, synthesis and processing, including multipliers. It provides robust and efficient implementations for a variety of frame-related operators for generic frames and several special types, e.g. Gabor and filter bank frames.  

In a recent release, reported in \cite{ltfatnote030}, a 'frames framework' was implemented, which models the abstract frame concept in an object-oriented approach. In this setting any algorithm can be designed to use a general frame. If a structured frame, e.g. of Gabor or wavelet type, is used, more efficient algorithms are automatically selected.

\section{\label{sec:erbfb} Filter bank frames: a signal processing viewpoint}

Linear time-invariant \emph{filter banks (FB)} are a classical signal analysis and processing tool. Their general, potentially non-uniform structure provides the natural setting for the design of flexible, frequency-adaptive time-frequency signal representations \cite{badokoto13}. In this section, we recall some basics of FB theory and consider the relation of perfect reconstruction FBs to certain frame systems. 

\subsection{Basics of filter banks}

In the following, we consider discrete signals with finite energy ($x\in\ell^2(\ZZ)$), interpreted as samples of a continuous signal, sampled at sampling frequency $\xi_s$, i.e. the signal was sampled every $1/\xi_s$ seconds. 
Bold italic letters indicate matrices (upper case), e.g. {$\bdi G$}, and vectors (lower case), e.g. {$\bdi h$}. We denote by $W_{N} = e^{2 i \pi /N}$ the $N$th root of unity and by $\delta_{k} = \delta_0[\cdot -k]$ the (discrete) Dirac symbol, with $\delta_k[n]=1$ for $n=k$ and $0$ otherwise. Observe that for $q = D/d$ we have 
\begin{equation}\label{eq:schah1}
	\sum_{l=0}^{q-1} W_{D}^{jld} = \sum_{l=0}^{q-1} e^{2\pi i jl/q} = \left\{%
	\begin{array}{ll}
		q & \text{if $j$ is a multiple of $q$}\\
		0 & \text{otherwise.}
	\end{array}\right.
\end{equation}

The \emph{$z$-transform} maps a \emph{(discrete-)time domain} signal $x$ to its \emph{frequency domain} representation $X$ by 
\[\mathcal{Z}: \, x[n] \mapsto X(z) = \sum_{n\in\ZZ} x[n]z^n \text{, for all } z\in\CC.\]
By setting $z = e^{2\pi i\xi}$ for $\xi \in \mathbb{T}$, the $z$-transform equals the discrete-time Fourier transform ($\mathrm{DTFT}$). Note that the $z$-transform is uniquely determined by its values on the complex unit circle~\cite{opsc89}. It is easy to see that, $\mathcal{Z} \left( \delta_{k} \right) = z^k$, a property that we will use later on.

The application of a filter to a signal $x$ is given by the convolution of $x$ with the time domain representation, or \emph{impulse response} $h\in\ell^2(\ZZ)$ of the filter 
\begin{equation}\label{eq:filtconv}
  y[n] = x\ast h[n] = \sum_{l\in\ZZ} x[l]h[n-l],\ \forall\ n\in\ZZ, 
\end{equation}
or equivalently by multiplication in the frequency domain $Y(z) = X(z)H(z)$, where $H(z)$ is the \emph{transfer function}, or frequency domain representation, of the filter.

Furthermore define the \emph{downsampling} and \emph{upsampling} operators $\downarrow_{d},\ \uparrow_{d}$ by 
\begin{equation}\label{eq:downupsample}
  \downarrow_{d}\left\{x\right\}[n] = x[d \cdot n]\quad \text{and}\quad \uparrow_{d}\left\{x\right\}[n] = \begin{cases}
    x[n/d] & \text{ if } n\in d\ZZ,\\
    0 & \text{ otherwise.}
    \end{cases}
\end{equation}
Here, $d\in\NN$ is called the \emph{downsampling} or \emph{upsampling factor}, respectively. In the frequency domain, the effect of down- and upsampling is the following~\cite{oppenheim1989discrete}:
\begin{equation}\label{eq:ZDUsample}
 \mathcal Z (\downarrow_{d}\left\{x\right\})(z) = d^{-1}\sum_{j=0}^{d-1} X(W_d^j z^{1/d})\quad \text{and}\quad \mathcal Z (\uparrow_{d}\left\{x\right\})(z) = X(z^d).
\end{equation}
In words, downsampling a signal by $d$ results in the dilation of its spectrum by $d$ and the addition of $(d-1)$ copies of the dilated spectrum. These copies of the spectrum (the terms $X(W_d^j z^{1/d})$ for $j \neq 0$ in the sum above) are called \textit{aliasing terms}. Conversely, upsampling a signal by $d$ results in the contraction of its spectrum by $d$.

An FB is a collection of analysis filters $H_{k}(z)$, synthesis filters $G_{k}(z)$, and downsampling and upsampling factors $d_k$, $k\in\{0,\ldots,K\}$, see Fig.~\ref{sfig:nonuniformFB}. 
An FB is called \emph{uniform}, if all filters have the same downsampling factor, i.e. $d_k = D$ for all $k$. 

\begin{figure}[!t]
	\begin{center}
	\includegraphics[width=0.89\textwidth]{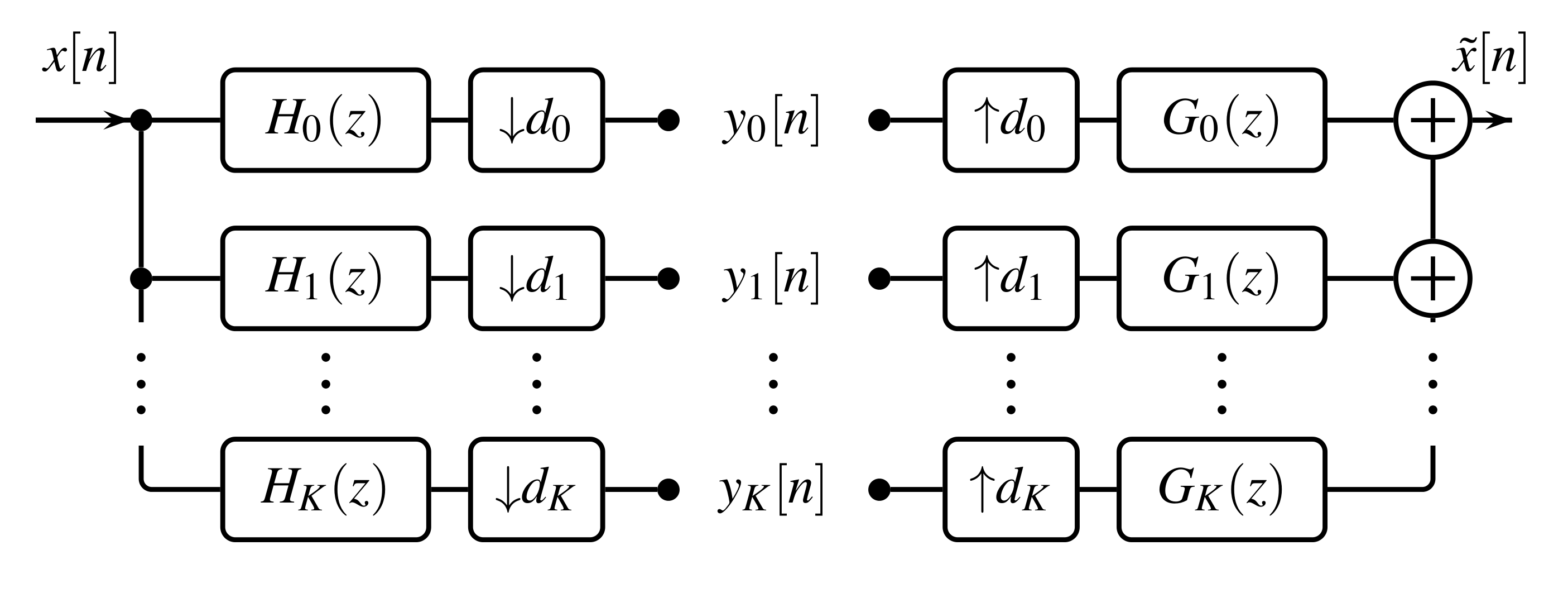}
	\end{center}
	\caption{General structure of a non-uniform analysis-synthesis FB.}
	\label{sfig:nonuniformFB}
\end{figure}

The sub-band components $y_{k}[n]$ of the system represented in Fig.~\ref{sfig:nonuniformFB} are given in the time domain by 
\begin{equation}\label{eq:yknonuniform}
	y_{k}[n] \, = \, \downarrow_{d_{k}}\left\{ h_{k}*x\right\}[n]
\end{equation}
The output signal is $\tilde{x}[n] \, =\, \sum_{k=0}^{K}\left(g_{k}*\uparrow_{d_{k}}\left\{y_{k}\right\} \right)[n]$. When analyzing the properties of a filter (bank), it is often useful to transform the expression for $\tilde{x}$ to the frequency domain. First, apply the z-transform to the output of a single analysis/synthesis branch, obtaining
\begin{equation} \label{eq:filterz0}
   \mathcal Z \left(g_{k}*\uparrow_{d_{k}}\left\{y_{k}\right\} \right)(z) = d_k^{-1}[X(W^0_{d_k}z),\ldots,X(W^{d_k-1}_{d_k}z)]\left[\begin{matrix}
                                                             H_k(W^0_{d_k}z)\\
                                                             \vdots\\
                                                             H_k(W^{d_k-1}_{d_k}z)
                                                           \end{matrix}\right] G_k(z),
\end{equation}
where the down- and upsampling properties of the z-transform were applied, see Eq. \eqref{eq:ZDUsample}. Now let $D = \lcm(d_0,\ldots,d_K)$, i.e. the least common multiple of the downsampling factors, and $D/d_k = q_k$. Then \eqref{eq:filterz0} gives 

\begin{equation}
   \mathcal Z \left(g_{k}*\uparrow_{d_{k}}\left\{y_{k}\right\} \right)(z) = D^{-1}[X(W^0_{D}z),\ldots,X(W^{D-1}_{D}z)]\bdi h_k(z) G_k(z)\label{eq:polyphaserep1channel},      \end{equation}
where , 
\[\bdi h_{k}(z) = q_k \cdot \Big[
	H_k(z), \underbrace{0 , \cdots , 0}_{q_{k}-1 \text{ zeros}}, 
	H_k(W_{D}^{q_k}z),  \underbrace{0 , \cdots , 0}_{q_{k}-1 \text{ zeros}},
	\cdots, 
	H_k(W_{D}^{(d_k-1)q_k}z), \underbrace{0 , \cdots , 0}_{q_{k}-1 \text{ zeros}}
	\Big]^T.\]
The relevance of this equality becomes clear if we use linearity of the z-transform to obtain a frequency domain representation of the full FB output, also called the \emph{alias domain representation} \cite{Vaidyanathan:1993a}
\begin{eqnarray}
    \tilde{X}(z) &=& \sum_{k=0}^K \mathcal Z \left(g_{k}*\uparrow_{d_{k}}\left\{y_{k}\right\} \right)(z)\nonumber\\
    &=& D^{-1}[X(W^0_{D}z),\ldots,X(W^{D-1}_D z)]\left[\bdi h_0(z),\ldots,\bdi h_K(z) \right]\left[\begin{matrix}
                                                                                            G_0(z)\\
                                                                                            \vdots\\
                                                                                            G_K(z)
                                                                                           \end{matrix}\right]\nonumber\\
    &=& D^{-1}[X(W^0_{D}z),\ldots,X(W^{D-1}_D z)]\bdi H(z) \bdi G(z)\label{eq:polyphaserep},
  \end{eqnarray}
where $\bdi H(z) = \left[\bdi h_0(z),\ldots,\bdi h_K(z) \right]$ is the $D\times (K+1)$ \emph{alias component matrix} \cite{Vaidyanathan:1993a} and $\bdi G(z) = \left[G_0(z),\ldots,G_K(z) \right]$.

An FB system is \emph{undersampled, critically sampled or oversampled}, if $R=\sum_{k=0}^K d_k^{-1}$ is smaller than, equal to or larger than $1$, respectively. Consequently, a uniform FB is critically sampled if it has exactly $D$ subbands. For a deeper treatment of FBs, see e.g. \cite{Kovacevic:1993a,Vaidyanathan:1993a}.

{\bf Perfect reconstruction FBs:} An FB is said to provide perfect reconstruction if $\tilde{x}[n] = x[n-l]$ for all $x\in\ell^2(\ZZ)$ and some fixed $l\in\ZZ$. In the case when $l\neq 0$, the FB output is \emph{delayed} by $l$. Using the alias domain representation of the FB, the \emph{perfect reconstruction condition} can be expressed as
\begin{equation}\label{eq:PRcondition}
	\bdi H(z) \, \bdi G(z)  = z^l\left[D \; 0 \cdots 0 \; \right]^{T},
\end{equation}
for some $l\in\ZZ$, as this condition is equivalent to $\tilde{X}(z) = z^l X(z) = \mathcal Z (x*\delta_k)(z)$.
From this vantage point the perfect reconstruction condition can be interpreted as all the alias components (i.e. from the $2$nd to $D+1$-th) in $\bdi H(z)$ being uniformly canceled over all $z\in\CC$ by the synthesis filters $\bdi G(z)$, while the first component of $\bdi H(z)$ remains constant over all $z\in\CC$ (up to a fixed power of $z$). The perfect reconstruction condition is of tremendous importance for determining whether an FB, including both analysis and synthesis steps, provides perfect reconstruction. However, given a fixed analysis FB, the alias domain representation may fail to provide straightforward or efficient ways to find 
suitable synthesis filters that provide perfect reconstruction. It can sometimes be used to determine whether such a system can exist, although the process is far from intuitive~\cite{Hoang:1989a}.
Consequently, non-uniform perfect reconstruction FBs are still not completely investigated, and thus frame theory may provide valuable new insights. However, for uniform FBs the perfect reconstruction conditions have been largely treated in the literature~\cite{Kovacevic:1993a,Vaidyanathan:1993a}. Therefore, before we indulge in the frame theory of FBs, we also show how a non-uniform FB can be decomposed into its equivalent uniform FB. Such a uniform equivalent of the FB always exists~\cite{Kovacevic:1993a,Akkarakaran:2003a} and can be obtained as shown in Fig.~\ref{sfig:equniformFB} and described below.

\subsection{The equivalent uniform filter bank}

\begin{figure}
	\begin{center}
	 \includegraphics[width=0.89\textwidth]{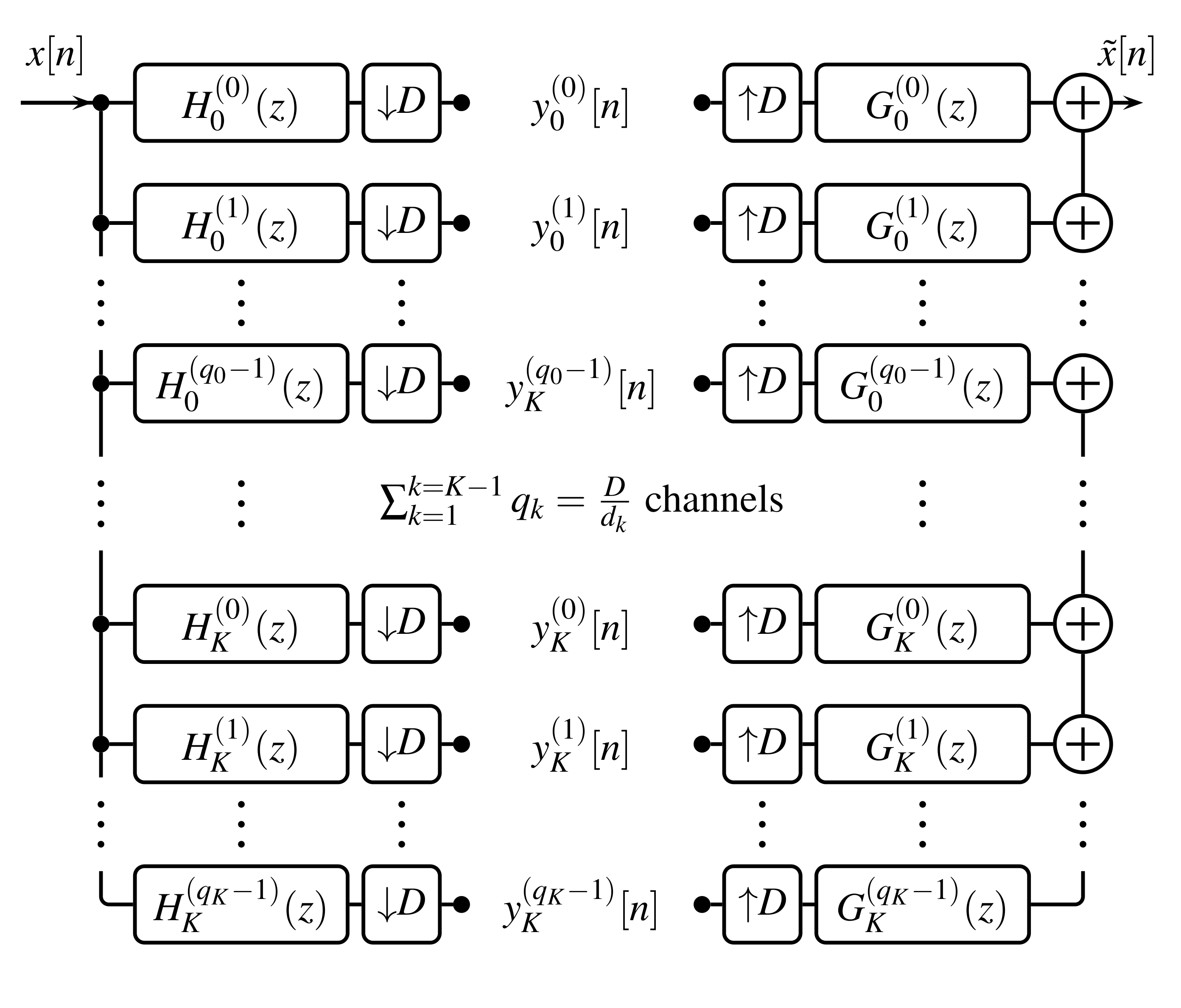}
	\end{center}
	\caption{The equivalent uniform FB~\cite{Akkarakaran:2003a} corresponding to the non-uniform FB in Fig.~\ref{sfig:nonuniformFB}. The terms $H_k^{(l)}$ and $G_k^{(l)}$ in~(b) correspond to the $z$-transforms of the terms $h_{k}^{(l)}$ and $g_{k}^{(l)}$ defined in~\eqref{eq:delayedfilts}.}
	\label{sfig:equniformFB}
\end{figure}

To construct the equivalent uniform FB to a general FB specified by analysis filters $H_{k}(z)$, synthesis filters $G_{k}(z)$, and downsampling and upsampling factors $d_k$, $k\in\{0,\ldots,K\}$, start by denoting again $D = \lcm(d_0,\ldots,d_K)$. We first construct the desired uniform FB, before showing that it is in fact equivalent to the given non-uniform FB. For every filter $h_k, g_k$ in the non-uniform FB, introduce $q_k = D/d_k$ filters, given by specific delayed versions of $h_k,g_k$:
\begin{equation}\label{eq:delayedfilts}
  h_k^{(l)}[n] = h_k\ast \delta_{ld_k} = h_k \left[ n - l d_k\right] \quad \text{and} \quad g_k^{(l)}[n] = g_k\ast \delta_{-ld_k} = g_k \left[ n + l d_k\right],
\end{equation}
for $l=0,\ldots,q_k-1$. It is easily seen that convolution with $\delta_k$ equals translation by $k$ samples by just checking the definition of the convolution operation \eqref{eq:filtconv}. Consequently, the sub-band components are 
\begin{equation}
	y_{k}^{(l)}[n] = y_{k}[n q_{k} - l] = \downarrow_{D}\{\underbrace{h_{k}*\delta_{ld_{k}}}_{:= h_{k}^{(l)}}* x\} [n],
\end{equation}
where $y_k$ is the $k$-th sub-band component with respect to the non-uniform FB. Thus, by grouping the corresponding $q_k$ sub-bands, we obtain
\[y_{k}[n] = \sum_{l=0}^{q_{k}-1}\uparrow_{q_{k}}\left\{y_{k}^{(l)}\right\} [n+l]. \] 

In the frequency domain, the filters $h_k^{(l)}, g_k^{(l)}$ are given by
\[
 H_k^{(l)}(z) = z^{ld_k}H_k(z) \quad \text{and} \quad G_k^{(l)}(z) = z^{-ld_k}G_k(z).
\]
Similar to before, the output of the FB can be written as
\begin{eqnarray}
	\widetilde{X}(z) &=& D^{-1}\sum_{k=0}^{K}\sum_{j=0}^{D-1}\sum_{l=0}^{q_{k}-1} G_{k}^{(l)}(z) H_{k}^{(l)}\left(W_{D}^{j}z\right)X\left(W_{D}^{j}z\right)\nonumber\\
	&=& D^{-1}\sum_{k=0}^{K}\sum_{j=0}^{D-1} G_{k}(z) H_{k}\left(W_{D}^{j}z\right)X\left(W_{D}^{j}z\right)\sum_{l=0}^{q_{k}-1} W_{D}^{jld_k}\label{eq:xtildez}
\end{eqnarray}
To obtain the second equality, we have used that $G_{k}^{(l)}(z) H_{k}^{(l)}\left(W_{D}^{j}z\right) = W_{D}^{jld_k}G_{k}(z)H_{k}\left(W_{D}^{jld_k}z\right)$.
Insert Eq. \eqref{eq:schah1} into \eqref{eq:xtildez} to obtain
\begin{eqnarray}
	\widetilde{X}(z) &=& D^{-1}\sum_{k=0}^{K}\sum_{j=0}^{d_k-1} q_k G_{k}(z) H_{k}\left(W_{D}^{jq_k}z\right)X\left(W_{D}^{jq_k}z\right)\nonumber\\
	&=& D^{-1}\sum_{k=0}^{K} [X(W_D^0 z),\ldots,X(W_D^{D_1}(z)]\bdi h_k(z)G_{k}(z)\nonumber\\
	&=& D^{-1}[X(W^0_{D}z),\ldots,X(W^{D-1}_D z)]\bdi H(z) \bdi G(z),	
\end{eqnarray}
which is exactly the output of the non-uniform FB specified by the $h_k$'s, $g_k$'s and $d_k$'s, see \eqref{eq:polyphaserep}. Therefore, we see that an equivalent uniform FB for every non-uniform FB is obtained by decomposing each $k$-th channel of the non-uniform system into $q_k$ channels. The uniform system then features $\sum_{k=0}^{K} q_k$ channels in total with the downsampling factor $D=\lcm(d_0,\ldots,d_K)$ in all channels. 

\subsection{\label{ssec:frametheory} Connection to Frame Theory}

We will now describe in detail the connection between non-uniform FBs and frame theory. The main difference to previous work in this direction, cf. \cite{Bolcskei:1998a,Cvetkovic:1998a,chai10,Fickus:2013a}, is that we do not restrict to the case of uniform FBs. The results in this section are not new, but this presentation is their first appearance in the context of non-uniform FBs. Besides using the equivalent uniform FB representation, see Fig. \ref{sfig:equniformFB}, we transfer results previously obtained for \emph{generalized shift-invariant systems}~\cite{rosh04,helawe02} and nonstationary Gabor systems~\cite{badohojave11,Holighaus:2013a,ho14-1} to the non-uniform FB setting. For that purpose, we consider frames over the Hilbert space $\mathcal{H} = \ell^2(\mathbb Z)$ of finite energy sequences. Moreover, we consider only FBs with a finite number $K+1\in\NN$ of channels, a setup naturally satisfied in every real-
world application. The central observation linking FBs to frames is that the convolution can be expressed as an inner product:
\[
 y_{k}[n] =  \downarrow_{d_{k}}\left\{ h_{k}*x\right\}[n] = \langle x, \overline{h_k[nd_k -\cdot]} \rangle
\]
where the bar denotes the complex conjugate. Hence, the sub-band components with respect to the filters $h_k$ and downsampling factors $d_k$ equal the frame coefficients of the system $\Phi = \left(\overline{h_k[nd_k -\cdot]}\right)_{k,n}$. Note that the upper frame inequality, see Eq. \eqref{frdef}, is equivalent to the $h_k$'s and $d_k$'s defining a system where bounded energy of the input implies bounded energy of the output.
We will investigate the frame properties of this system by transference to the Fourier domain \cite{bacahemo11};
we consider $\widehat{\Phi} =\left(\bd E_{-nd_k}\widehat{h_k}\right)_{k,n}$, where $\widehat{h_k}(\xi) = \overline{H_k(e^{2\pi i\xi})}$ denotes the Fourier transform of $\overline{h_k[-\cdot]}$ and the operator $\bd E_{\omega}$ denotes modulation, i.e. $\bd E_{-nd_k}\widehat{h_k}(\xi) = \widehat{h_k}(\xi)e^{-2\pi ind_k \xi}$. 

If $\Phi$ satisfies at least the upper frame inequality in Eq. \eqref{frdef}, then the frame operators $\bd S_\Phi$ and $\bd S_{\widehat{\Phi}}$ are related by the matrix Fourier transform~\cite{ba05-1}: 
\[
 \bd S_{\widehat{\Phi}} = \mathcal{F}_{DT}\,\bd S_\Phi\,\mathcal{F}_{DT}^{-1},
\]
where $\mathcal{F}_{DT}$ denotes the discrete-time Fourier transform. Since the matrix Fourier transform is a unitary operation, the study of the frame properties of $\Phi$ reduces to the study of the operator $\bd S_{\widehat{\Phi}}$. In the context of FBs, the frame operator can be expressed as the action of an FB with analysis filters $h_k$'s, downsampling and upsampling factors $d_k$'s, and synthesis filters $\overline{h_k[-\cdot]}$. That is, the synthesis filters are given by the time-reversed, conjugate impulse responses of the analysis filters. This is a very common approach to FB synthesis. But note that it only gives perfect reconstruction if the system constitutes a Parseval frame, see Prop. \ref{sec:parseval1}. The z-transform of a time-reversed, conjugated signal is given by $\mathcal{Z}(\overline{h[-\cdot]})(z) = \overline{\mathcal{Z}(h)(1/\overline{z})}$. Inserting this into the alias domain representation of the FB \eqref{eq:polyphaserep} yields
\begin{eqnarray}
  \bd S_{\widehat{\Phi}}X(z) & = \frac{1}{D}\left[X(W_{D}^{0}z) \cdots X(W_{D}^{D-1}z)\right] \bdi H(z) \left[%
      	\begin{array}{c}
        	\overline{H_0(1/\overline{z})}\\
        	\vdots \\
        	\overline{H_K(1/\overline{z})}
        \end{array}\right]
\end{eqnarray}
or, restricted to the Fourier domain
\begin{equation}\label{eq:poly2}
  \bd S_{\widehat{\Phi}}X(e^{2\pi i\xi}) = [X(e^{2\pi i(\xi+0/D)}) \cdots X(e^{2\pi i(\xi+(D-1)/D)})]\mathcal H(\xi), 
\end{equation}
with  
\begin{equation}\label{eq:poly3}
	\mathcal H(\xi) := [\mathcal{H}_0(\xi),\ldots,\mathcal{H}_{D-1}(\xi)]^T := \frac{1}{D}\bdi H(e^{2\pi i\xi})\left[\overline{H_0(e^{2\pi i\xi})},\ldots,\overline{H_K(e^{2\pi i\xi})}\right]^T,
\end{equation}
for $\xi\in\mathbb{T} = \RR/\ZZ$. Here, we used $\overline{1/e^{2\pi i \omega}} = e^{2\pi i\omega}$ for all $\omega\in\RR$. We call $\mathcal H_0$ the \emph{frequency response} and $\mathcal H_n$, $n=1,\cdot D-1$ the \emph{alias components} of the FB.

Another way to derive Eq. \eqref{eq:poly2} is by using the 
the Walnut representation of the frame operator for the nonstationary Gabor frame $\widehat{\Phi} = \left(\bd E_{-nd_k}\widehat{h_k}\right)_{k,n}$, first introduced in \cite{DBLP:journals/ijwmip/DorflerM14} for the continuous case setting.

\begin{Pro}\label{pro:walnut}
  Let $\widehat{\Phi} = \left(\bd E_{-nd_k}\widehat{h_k}\right)_{k\in\{0,\ldots,K\},n\in\ZZ}$, with $\widehat{h_k}\in L^2(\TT)$ being (essentially) bounded
   and $d_k\in\NN$. Then the frame operator $\bd S_{\widehat{\Phi}}$ admits the Walnut representation 
  \begin{equation}\label{eq:walnut}
    S_{\widehat{\Phi}}\widehat{x}(\xi) = \sum_{k=0}^K\sum_{n=0}^{d_k-1} d_k^{-1}\widehat{h_k}(\xi)\overline{\widehat{h_k}(\xi-nd_k^{-1})}\widehat{x}(\xi-nd_k^{-1}),
  \end{equation}
  for almost every $\xi\in\TT$ and all $\widehat{x}\in L^2(\TT)$.
\end{Pro}
\begin{proof}
  By the definition of the frame operator, see Eq. \eqref{eq:frameop},  we have
  $$\bd S_{\widehat{\Phi}} \widehat{x}(\xi) = \sum \limits_{k,n} \left< \widehat{x} , \widehat{h_k} e^{- 2 \pi i n d_k \xi} \right> \widehat{h_k}(\xi) e^{-2 \pi i n d_k \xi}.$$
  Note that 
  $$\sum_{n\in\ZZ} \langle \widehat{x},e^{-2\pi i\xi n d_k} \widehat{h_k}\rangle e^{-2\pi i\xi n d_k} = \sum_{n\in\ZZ} \mathcal{F}_{DT}^{-1}(\widehat{x}\overline{\widehat{h_k}})[nd_k] e^{-2\pi i\xi nd_k}. $$
  to get the result by applying Poisson's summation  formula, see e.g. \cite{gr01}.
\end{proof}

The sums in \eqref{eq:walnut} can be reordered to obtain
\[
  \sum_{n=0}^{D-1} \widehat{x}(\xi-nD^{-1})\sum_{k\in K_n} d_k^{-1}\widehat{h_k}(\xi)\overline{\widehat{h_k}(\xi-nD^{-1})}, 
\]
where $K_n = \{k \in \{0,\ldots,K\}~:~ nD^{-1}= jd_k^{-1} \text{ for some } j\in\NN \}$. Inserting $\widehat{h_k}(\xi) = \overline{H_k(e^{2\pi i\xi})}$ and comparing the definition of $\mathcal H_n$ in \eqref{eq:poly3}, we can see that 
\[ 
 \sum_{k\in K_n} \widehat{h_k}(\xi)\overline{\widehat{h_k}(\xi-nD^{-1})}
 = \sum_{k\in K_n} \overline{H_k(e^{2\pi i\xi})}H_k(e^{2\pi i(\xi-n/D^{-1})}) = \mathcal{H}_{n}(\xi)
\]
for almost every $\xi\in\TT$ and all $n=0,\ldots,D-1$. Hence, we recover the representation of the frame operator as per \eqref{eq:poly2}, as expected. What makes Proposition \ref{pro:walnut} so interesting, is that it facilitates the derivation of some important sufficient frame conditions. The first is a generalization of the theory of painless non-orthogonal expansions by Daubechies et al.~\cite{dagrme86}, see also~\cite{badohojave11} for a direct proof.

\begin{Cor}
  Let $\widehat{\Phi} = \left(\bd E_{-nd_k}\widehat{h_k}\right)_{k\in\{0,\ldots,K\},n\in\ZZ}$, with $\widehat{h_k}\in L^2(\TT)$ and $d_k\in\NN$. Assume for all $0 \leq k \leq K$, there is $I_k \subseteq \mathbb{T}$ with $|I_k|\leq d_k^{-1}$ and $\widehat{h_k}(\xi) = 0$ for almost every $\xi\in \mathbb{T}\setminus I_k$. Then $\widehat{\Phi}$ is a frame if and only if there are $A,B$ such that
  \begin{equation}\label{eq:maindiag}
    0 < A \leq \sum_{k=0}^K d_k^{-1}|\widehat{h_k}|^2 = \mathcal{H}_0 \leq B < \infty, \text{ a.e. }
  \end{equation}
  Moreover, a dual frame for $\widehat{\Phi}$ is given by $\widehat{\Psi} = \left(\bd E_{-nd_k}\widehat{g_k}\right)_{k\in\{0,\ldots,K\},n\in\ZZ}$, where
  \begin{equation}\label{eq:candualPL}
    \widehat{g_k}(\xi) = \frac{\widehat{h_k}(\xi)}{\mathcal{H}_0(\xi)} \text{ a.e.}
  \end{equation}
\end{Cor}
\begin{proof}
  First, note that the existence of the upper bound $B$ is equivalent to $\widehat{h_k}\in L^\infty(\TT)$, for all $k = 0,\ldots,K$. It is easy to see that under the assumptions given, Eq. \eqref{eq:walnut} equals
  \[
   \bd S_{\widehat{\Phi}}\widehat{x}(\xi) = \widehat{x}(\xi)\sum_{k=0}^K d_k^{-1}|\widehat{h_k}|^2(\xi) = \widehat{x}(\xi) \cdot \mathcal{H}_0(\xi).
  \]
  Hence, $\bd S_{\widehat{\Phi}}$ is invertible if and only if $\mathcal{H}_0$ is bounded above and below, proving the first part. Moreover, $\bd S_{\widehat{\Phi}}^{-1}$ is given by pointwise multiplication with $1/\mathcal{H}_0$ and therefore, the elements of the canonical dual frame for $\widehat{\Phi}$, defined in Eq. \eqref{eq:candual}, are given by
  \[
   \bd S_{\widehat{\Phi}}^{-1}\bd E_{-nd_k}\widehat{h_k} = \frac{\bd E_{-nd_k}\widehat{h_k}}{\mathcal{H}_0} = \bd E_{-nd_k}\frac{\widehat{h_k}}{\mathcal{H}_0} = \widehat{g_k}.
  \]
\end{proof}

In other words, recalling $\widehat{h_k}(\xi) = \overline{H_k(e^{2\pi i\xi})}$, if the filters $h_k$ are strictly band-limited, the downsampling factors $d_k$ are small and $0<A\leq \mathcal{H}_0\leq B< 0$ almost everywhere, then we obtain a perfect reconstruction system with synthesis filters $g_k$ defined by
\[
 G_k(e^{2\pi i\xi}) = \frac{\overline{H_k(e^{2\pi i\xi})}}{\mathcal{H}_0(\xi)}.
\]

The second, more general and more interesting condition can be likened to a diagonal dominance result, i.e. if the main term $\mathcal{H}_0$ is \emph{stronger} than the sum of the magnitude of alias components $\mathcal{H}_n$, $n= 1,\ldots,D-1$, then the FB analysis provided by the filters $h_k$ and downsampling factors $d_k$ is invertible. 

\begin{Pro}\label{pro:diagdom}
  Let $\widehat{\Phi} = \left(\bd E_{-nd_k}\widehat{h_k}\right)_{k\in\{0,\ldots,K\},n\in\ZZ}$, with $\widehat{h_k}\in L^2(\TT)$ and $d_k\in\NN$. If there are $0< A \leq B < \infty$ with 
  \begin{equation}\label{eq:diagdom1}
    A \leq \sum_{k=0}^K d_k^{-1} |\widehat{h_k}|^2(\xi) \pm
    \sum_{k=0}^K\sum_{n=1}^{d_k-1} d_k^{-1}\left|\widehat{h_k}(\xi)\widehat{h_k}(\xi-nd_k^{-1})\right|
    \leq B,
  \end{equation}
  for almost every $\xi\in\TT$, then $\widehat{\Phi}$ forms a frame with frame bounds $A,B$.
\end{Pro}

Note that \eqref{eq:diagdom1} impliest $\widehat{h_k}\in \bd L^\infty(\RR)$ for all $k\in\{0,\ldots,K\}$. Therefore, Proposition \ref{pro:walnut} applies for any FB that satisfies \eqref{eq:diagdom1}.
The proof of Proposition \ref{pro:diagdom} is somewhat lengthy and we omit it here. It is very similar to the  proof of the analogous conditions for Gabor and wavelet frames that can be found in \cite{da92} for the continuous case. It can also be seen as a corollary of \cite[Theorem 3.4]{chgo15-1}, covering a more general setting. A few things should be noted regarding Proposition \ref{pro:diagdom}. 

(a) As mentioned before, this is a sort of diagonal dominance result. While the sum $\sum_{k=0}^K d_k^{-1} |\widehat{h_k}|^2(\xi)$ corresponds to $\mathcal{H}_0$, we have 
\[
 \sum_{k=0}^K\sum_{n=1}^{d_k-1} d_k^{-1}\left|\widehat{h_k}(\xi)\widehat{h_k}(\xi-nd_k^{-1})\right| = \sum_{n=1}^{D-1} |\mathcal{H}_n|(\xi).
\]
Since, in fact, the finite number of channels guarantees the existence of $B$ if and only if $\widehat{h_k}\in L^\infty(\TT)$, for all $k = 0,\ldots,K$, the result implies that the FB analysis provided by $h_k$'s and $d_k$'s is invertible, whenever
\[
 \mathcal{H}_0 - \sum_{n=1}^{D-1} |\mathcal{H}_n| \geq A > 0, \text{ almost everywhere.}
\]

(b) No explicit dual frame is provided by Proposition \ref{pro:diagdom}. So, while we can determine invertibility quite easily, provided the Fourier transforms of the filters can be computed, the actual inversion process is still up in the air. In fact, it is unclear whether there are synthesis filters $g_k$ such that the $h_k$'s and $g_k$'s form a perfect reconstruction system with down-/upsampling factors $d_k$. We consider here two possible means of recovering the original signal $X$ from the sub-band components $Y_k$. 

First, the equivalent unform FB, comprised of the filters $h^{(l)}_{k}$, for $l\in\{0,\ldots,q_k-1\}$ and all $k\in \{0,\cdots,K\}$, with downsampling factor $D=\lcm(d_k~:~k\in\{0,\ldots,K\})$ can be constructed. Since the non-uniform FB forms a frame, so does its uniform equivalent and hence the existence of a dual FB $g^{(l)}_{k}$, for $l\in\{0,\ldots,q_k-1\}$ and all $k\in \{0,\cdots,K\}$, is guaranteed. Note that the $g^{(l)}_{k}$ are not necessarily delayed versions of $g^{(0)}_{k}$, as it is the case for $h^{(l)}_{k}$. Then, the structure of the alias domain representation in \eqref{eq:polyphaserep} with $g_k = \overline{h_k[-\cdot]}$ can be exploited~\cite{Bolcskei:1998a} to obtain perfect reconstruction synthesis. In the finite, discrete setting, i.e. when considering signals in $\RR^L$ ($\CC^L$), a dual FB can be computed explicitly and efficiently by a generalization of the methods presented by Strohmer~\cite{st98-1}, see also~\cite{ltfatnote038}. In practice, 
both the storage and time efficiency of computing the dual uniform FB rely crucially on $D=\lcm(d_k~:~k\
in\{0,\ldots,K\})$ 
being small, i.e. $\sum_k q_k$ not being much larger than $K+1$.

If that is not the case, the frame property of $\widehat{\Phi} = \left(\bd E_{-nd_k}\widehat{h_k}\right)_{k\in\{0,\ldots,K\},n\in\ZZ}$ guarantees the convergence of the Neumann series 
\begin{equation}
  \bd S_{\widehat{\Phi}}^{-1} = \frac{2}{A_0+B_0}\sum_{l=0}^\infty \left(\bd I - \frac{2}{A_0+B_0} \bd S_{\widehat{\Phi}}\right)^l,
\end{equation}
where $0<A_0\leq B_0<\infty$ are the optimal frame bounds of $\widehat{\Phi}$.
Instead of computing the elements of any dual frame explicitly, we can apply the inverse frame operator to the FB output
\begin{equation}\label{eq:CG}
	\tilde{X}(z) = \bd S_{\widehat{\Phi}}X(z) = \sum_{k=0}^{K} Y_k(z^{d_{k}}) H_k(z),
\end{equation}
obtaining $\bd S_{\widehat{\Phi}}^{-1}\tilde{X}=X$. This can be implemented with the \emph{frame algorithm}~\cite{dusc52,Grochenig:1993a}. However, any frame operator is positive definite and self-adjoint, allowing  for extremely efficient implementation via the \emph{conjugate gradients (CG)}~\cite{Grochenig:1993a,Trefethen:1997a} algorithm. In addition to a significant boost in efficiency compared to the frame algorithm, the conjugate gradients algorithm does not require an estimate of the optimal frame bounds $A_0,B_0$ and convergence speed depends solely on the condition number of $\bd S_{\widehat{\Phi}}$. It provides guaranteed, exact convergence in $L$ steps for signals in $\CC^L$, where every step essentially comprises one analysis and one synthesis step with the filters $h_k$ and $g_k = \overline{h_k[-\cdot]}$, respectively. If furthermore, $\mathcal{H}_0 \gg \sum_{n=1}^{D-1} |\mathcal{H}_n|$, then convergence speed can be further increased by preconditioning \cite{Balazs:2006a}, considering instead 
the operator defined by 
\[
 \widetilde{\bd S_{\widehat{\Phi}}}X(e^{2\pi i \xi}) = \mathcal{H}_0(\xi)^{-1}\bd S_{\widehat{\Phi}}X(e^{2\pi i \xi}).
\]
More specifically, the CG algorithm is employed to solve the system $\bd D_\Phi c = \bd S_\Phi x$ for $x$, given the coefficients $c$. 
Recall the analysis/synthesis operators $\bd C_\Phi,\bd D_\Phi$ (see Sec. \ref{froperators}), associated to
 a frame $\Phi$, which are equivalent to the analysis/synthesis stages of the FB. The preconditioned case can be implemented most efficiently, by precomputing an approximate dual FB, defined by $G_k(e^{2\pi i \xi}) = \mathcal{H}_0(\xi)^{-1}H_k(e^{2\pi i \xi})$ and solving instead
\[
 \bd D_\Psi c = \mathcal F^{-1} \mathcal{H}_0(\xi)^{-1}\bd S_{\widehat{\Phi}}\mathcal F x = \bd D_\Psi \bd C_{\Phi} x,\ \text{where } \Psi = \{\overline{g_k[nd_k-\cdot]}\}_{k,n},
\]
for $x$, given the coefficients $c$. Algorithm \ref{alg:nsgsyniter} shows a pseudo-code implementation of such a preconditioned CG scheme, available in the LTFAT Toolbox as the routine \texttt{ifilterbankiter}.

\begin{algorithm} \small
	\caption{Iterative synthesis: $\tilde{x} = \textbf{FBSYN}^{it}(c,(h_k,g_k,d_k)_k,\lambda)$}\label{alg:nsgsyniter}
	\begin{algorithmic}[1]
		\State Initialize $x_0 = 0$, $k = 0$
		\State $b \gets \bd D_\Psi c$
		\State $r_0\gets b$
		\State $h_0,p_0 \gets r_0$
		\Repeat
			\State $q_k = \bd D_\Psi (\bd C_\Phi p_0)$
			\State $\alpha_k\gets \frac{\langle r_k, h_k \rangle}{\langle p_k,q_k\rangle}$
			\State $x_{k+1}\gets x_k + \alpha_kp_k$
			\State $r_{k+1}\gets r_k + \alpha_kq_k$
			\State $h_{k+1}\gets r_{k+1}$
			\State $\beta_k\gets \frac{\langle r_{k+1}, h_{k+1} \rangle}{\langle r_k,h_k\rangle}$
			\State $p_{k+1}\gets h_{k+1} + \beta_kp_k$
			\State $k\gets k+1$
		\Until{$r_k \leq \lambda$}
		\State $\tilde{x}\gets x_k$
	\end{algorithmic}
\end{algorithm}

\section{\label{sec:appli}Frame Theory: Psychoacoustics-motivated Applications}

\subsection{A perfectly invertible, perceptually-motivated filter bank}\label{sec:audlet0}

The concept of auditory filters lends itself nicely to the implementation as a FB. 
As motivated in Sec.~\ref{Sec:intro0}, it can be expected that many audio signal processing applications greatly benefit from an invertible FB representation adapted to the auditory time-frequency resolution. Despite the auditory system showing significant nonlinear behavior, the results obtained through a linear representation are desirable for being much more predictable than when accounting for nonlinear effects. We call such a system \emph{perceptually-motivated FB}, to distinguish from \emph{auditory FBs} that attempt to mimic the nonlinearities in the auditory system. 
Note that, as mentioned in Section \ref{ssec:audfilters}, the first step in many auditory FBs is the computation of a perceptually-motivated FB, see e.g.~\cite{Irino:2006b}. The \emph{AUDlet FBs} we present here are a family of perceptually-motivated FBs that satisfy a perfect reconstruction property, offer flexible redundancy and enable efficient implementation. They 
were introduced in~\cite{bahoneso13,nebahopr15} and an implementation is available in the LTFAT Toolbox~\cite{basoto12}.

The AUDlet FB has a general non-uniform structure as
presented in Fig.~\ref{sfig:nonuniformFB} with analysis filters $H_{k}(z)$, synthesis filters $G_{k}(z)$, and downsampling and upsampling factors $d_k$.
Considering only real-valued signals allows us to deal with symmetric 
$\mathcal{F}_{DT}$s and  process only the positive-frequency range. Therefore let $K$ denote the number of filters in the frequency range $[f_{\mathrm{min}},f_{\mathrm{max}}]\cap [0,f_{s}/2[$, where $f_{\mathrm{min}} \geq 0 $ to $f_{\mathrm{max}} \leq f_{s}/2$ and $f_{s}/2$ is the Nyquist frequency, i.e. half the sampling frequency. If $f_{\mathrm{min}} > 0$, this range includes an additional filter at the zero frequency. Furthermore, another filter is always positioned at the Nyquist frequency to ensure that the full frequency range is covered. Thus, all FBs below feature $K+1$ filters in total and their redundancy is given by $R = d_{0}^{-
1} + 2 \sum_{k = 1}^{K-1} d_{k}^{-1} + d_{K}^{-1}$, since coefficients in the $1$st to $K-1$-th subbands are complex-valued.

\renewcommand{\arraystretch}{1.4}
\setlength{\tabcolsep}{.15cm}
\begin{table}[t]
\begin{center}
  \begin{tabular}{|c|c|c|}
    \hline
      Parameter & Role & Information \\
    \hline      
      $f_{\mathrm{min}}$ & minimum frequency in Hz & $f_{\mathrm{min}}\in [0,f_s/2[,\ f_{\mathrm{min}}<f_{\mathrm{max}}$\\
      $f_{\mathrm{max}}$ & maximum frequency in Hz & $f_{\mathrm{max}}\in ]0,f_s/2[,\ f_{\mathrm{max}}>f_{\mathrm{min}}$\\
      $f_k$ & center frequencies in Hz & $ F_\mathrm{AUD}^{-1}(F_\mathrm{AUD}(f_0)+k/V)$\\
      $K$ & (essential) number of channels & $K = V \left( F_\mathrm{AUD}(\xi_\mathrm{max}) - F_\mathrm{AUD}(f_\mathrm{min}) \right)+(1-\delta_{0,f_\mathrm{min}})$\\
      $V$ & channels per scale unit & $V = \left( F_\mathrm{AUD}(f_{k+1}) - F_\mathrm{AUD}(f_{k}) \right)^{-1}$, $k\in[1,K-2]$\\
      $w$ & frequency domain filter prototype & $w\in L^2(\TT)$ \\
      $\Gamma_k$ & dilation factors & $r_{bw}BW_{\mathrm{AUD}}(f_k)$, $r_{bw}>0$ (default $=1$)\\
      
      $H_k$ & filter transfer functions & $H_k(e^{2i \pi\xi}) = \Gamma_k^{-\frac{1}{2}} w \left( \frac{f_s\cdot\xi - f_k}{\Gamma_k} \right)$ \\
      $d_k$ & downsampling factors & $r_{d}BW_{\mathrm{AUD}}^{-1}(\xi_k)$, $r_{d}>0$ (default non-uniform = 1)\\
      $R$ & redundancy & $R = d_{0}^{-1} + 2 \sum_{k = 1}^{K-1} d_{k}^{-1} + d_{K}^{-1}$\\
    \hline   
  \end{tabular}
  \caption{Parameters of the perceptually-motivated AUDlet FB} 
  \end{center}\label{tab:1}
\end{table}

The AUDlet filters $H_k$'s, $k \in \{0,\ldots,K\}$ are constructed in the frequency domain by
\begin{equation}\label{eq:erbfilters}
	H_k(e^{2i \pi\xi}) = \Gamma_k^{-\frac{1}{2}} w \left( \frac{f_s\cdot\xi - f_k}{\Gamma_k} \right)
\end{equation}
where $w(\xi)$ is a prototype filter shape with bandwidth $1$ and center frequency $0$. Here, the shape factor $\Gamma_k$ controls the effective bandwidth of $H_k$ and $f_k$ determines its center frequency. The factor $\Gamma_k^{-1/2}$ ensures that all filters (i.e. for all $k$) have the same energy. To obtain filters equidistantly spaced on a perceptual frequency scale, the sets $\{f_k\}$ and $\{\Gamma_k\}$ are calculated using the corresponding $F_\mathrm{AUD}$ and $BW_{\mathrm{AUD}}$ formulas, see Tab. 1 
for more information on the AUDlet parameters and their relations. Since we emphasize inversion, the default analysis parameters are chosen such that the filters $H_k$ and downsampling factors $d_k$ form a frame. As an example, the AUDlet (a) and gammatone (b) analyses of a speech signal are represented in Fig.~\ref{fig:erbvsgfb_img} using AUD = ERB and $V$ = 6~filters per ERB. The filter prototype $w$ for the AUDlet was a Hann window. It can be seen that the two signal representations are very similar over the whole time-frequency plane. Since the gammatone filter is an acknowledged auditory filter model, this indicates that the time-frequency resolution of the AUDlet approximates well the auditory resolution.

\begin{figure*}[!t]
	\begin{center}
	 \includegraphics[width=0.49\columnwidth]{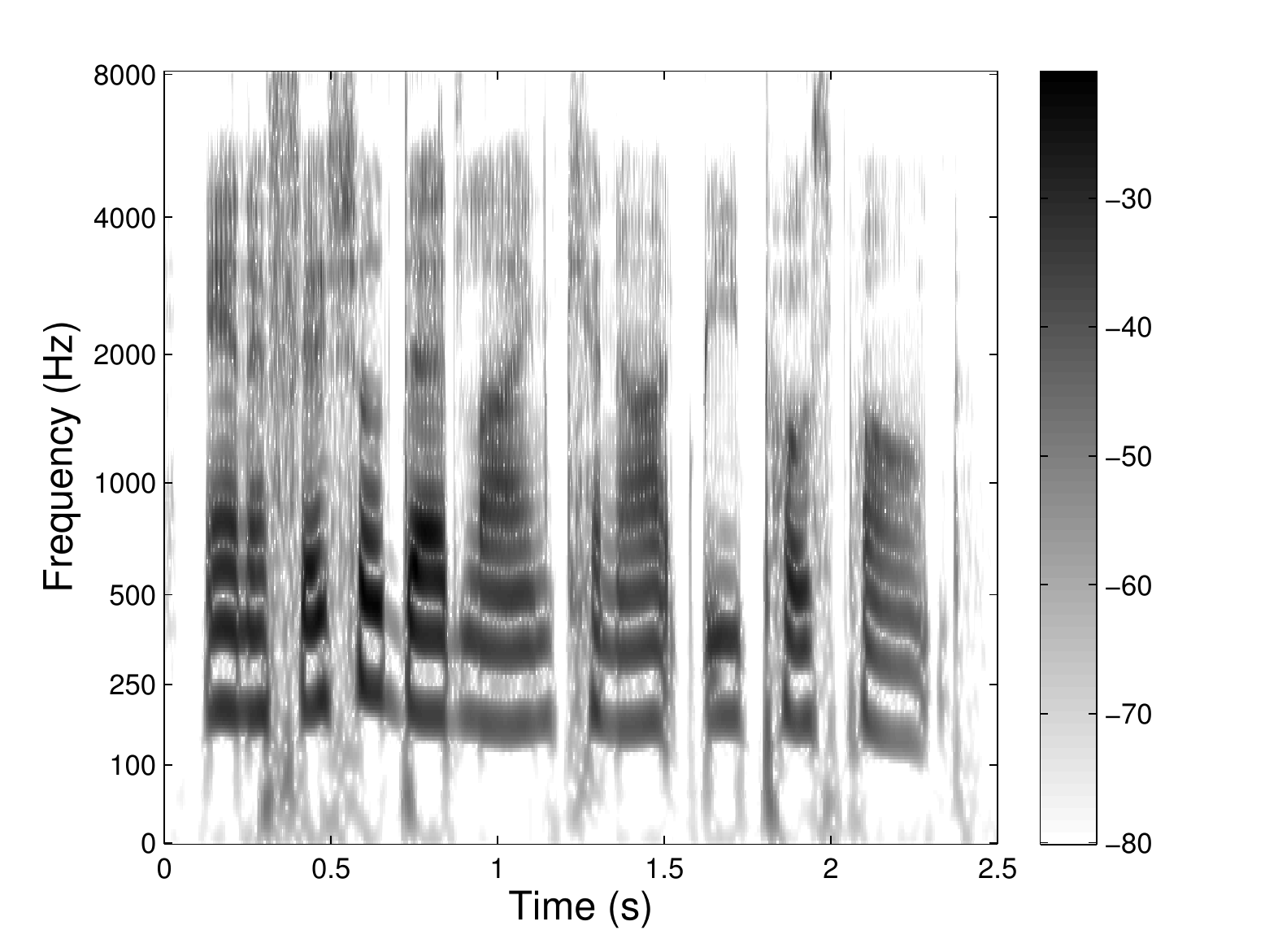}
	 \includegraphics[width=0.49\columnwidth]{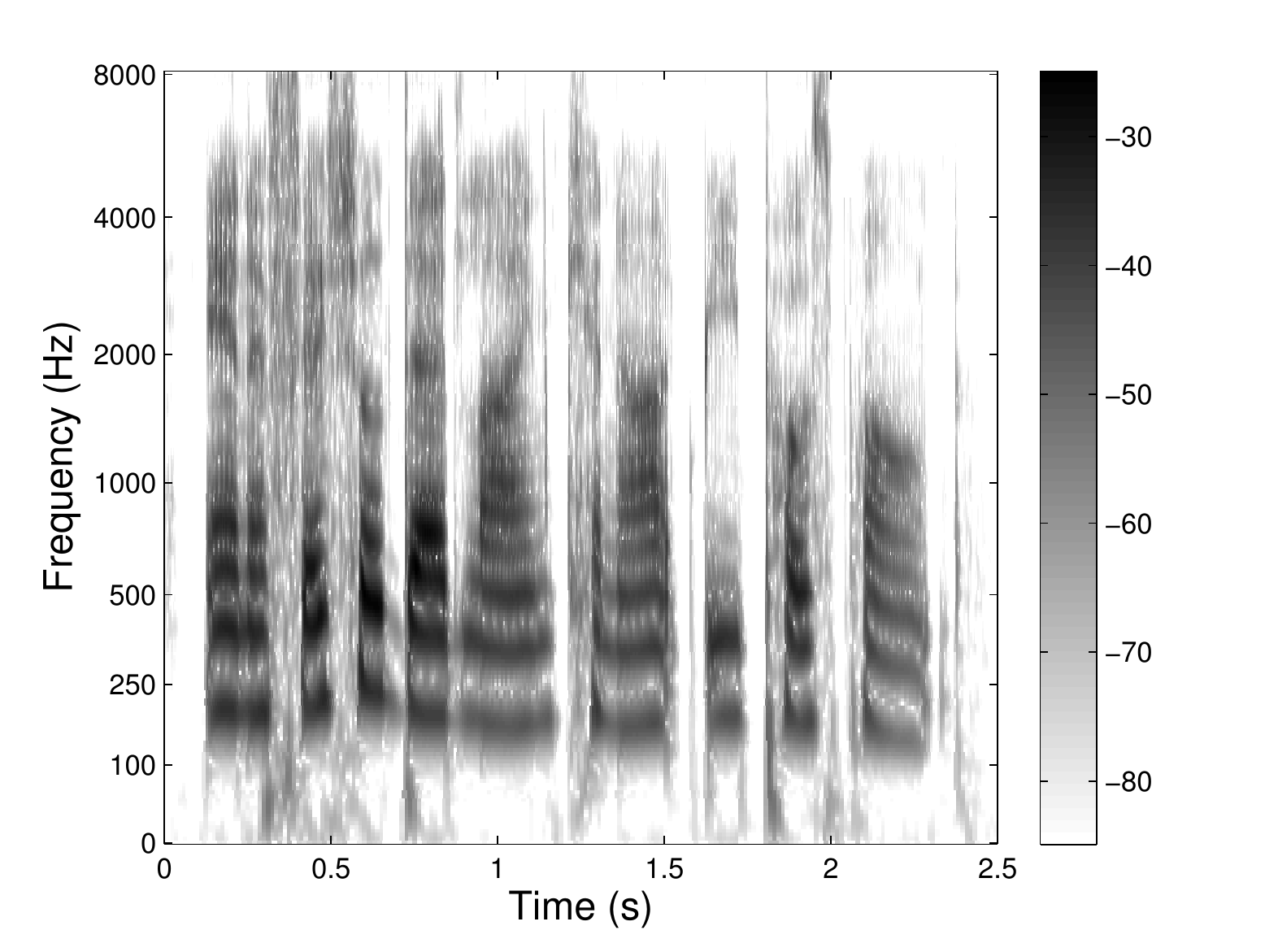}
	 \end{center}
	\caption{Analyses of a female speech signal taken from the TIMIT database \cite{Garofolo:1993a} by (a)~the AUDlet FB and~(b)~the gammatone FB using $V$ = 6~filters per ERB ($K$ = 201). It can be seen that
the two signal representations are very similar over the whole time-frequency plane.}
	\label{fig:erbvsgfb_img}  
\end{figure*}

\subsection{\label{sec:Irrel0}Perceptual Sparsity}

As discussed in Sec. \ref{ssec:masking} not all components of a sound perceived. This effect can be described by masking models and naturally leads to the following question: Given a time-frequency representation or any representation linked to audio, how can we apply that knowledge to only include audible coefficients in the synthesis? In an attempt to answer this question, efforts were made to combine frame theory and masking models into a concept called the \emph{Irrelevance Filter}. This concept is somehow linked to the currently very prominent sparsity and compressed sensing approach, see e.g. \cite{fo10-1,el10} for an overview. To reduce the amount of non-zero coefficients, the irrelevance filter uses a perceptual measure of sparsity, hence \emph{perceptual sparsity}. Perceptual and compressed sparsity can certainly be combined, see e.g. \cite{gineba14}. Similar to the methods used in compressed sensing, a redundant representation offers an advantage for perceptual sparsity, as well, as the same signal 
can 
be 
reconstructed from several sets of coefficients.

The concept of the irrelevance filter was first introduced in \cite{eck} and fully developed in \cite{Balazs:2010a}. It consists in removing the inaudible atoms in a Gabor transform while causing no audible difference to the original sound after re-synthesis. Precisely, an adaptive threshold function is calculated for each spectrum (i.e. at each time slice) of the Gabor transform using a simple model of spectral masking (see Sec. \ref{sec:SimMask0}), resulting in the so-called irrelevance threshold. Then, the amplitudes of all atoms falling below the irrelevance threshold are set to zero and the inverse transform is applied to the set of modified Gabor coefficients. This corresponds to an adaptive {\em Gabor frame multiplier} with coefficients in $\{0,1\}$. The application of the irrelevance filter to a musical signal sampled at 16~kHz is shown in Fig.~\ref{fig:irrelfilter}. A Matlab implementation of the algorithm proposed in \cite{Balazs:2010a} was used. All Gabor transform and filter 
parameters were identical to those mentioned in \cite{
Balazs:2010a}. Noteworthy, the offset parameter $o$ was set to -2.59~dB. In this particular example, about 48\% components were removed without causing any audible difference to the original sound after re-synthesis (as judged by informal listening by the authors). A formal listening test performed in \cite{Balazs:2010a} with 36 normal-hearing listeners and various musical and speech signals indicated that, on average, 36\% coefficients can be removed without causing any audible artifact in the re-synthesis.

\begin{figure}[!t]
	  \begin{center}
		 \subfloat[Original Gabor transform]{
		 	\includegraphics[width=0.46\textwidth]{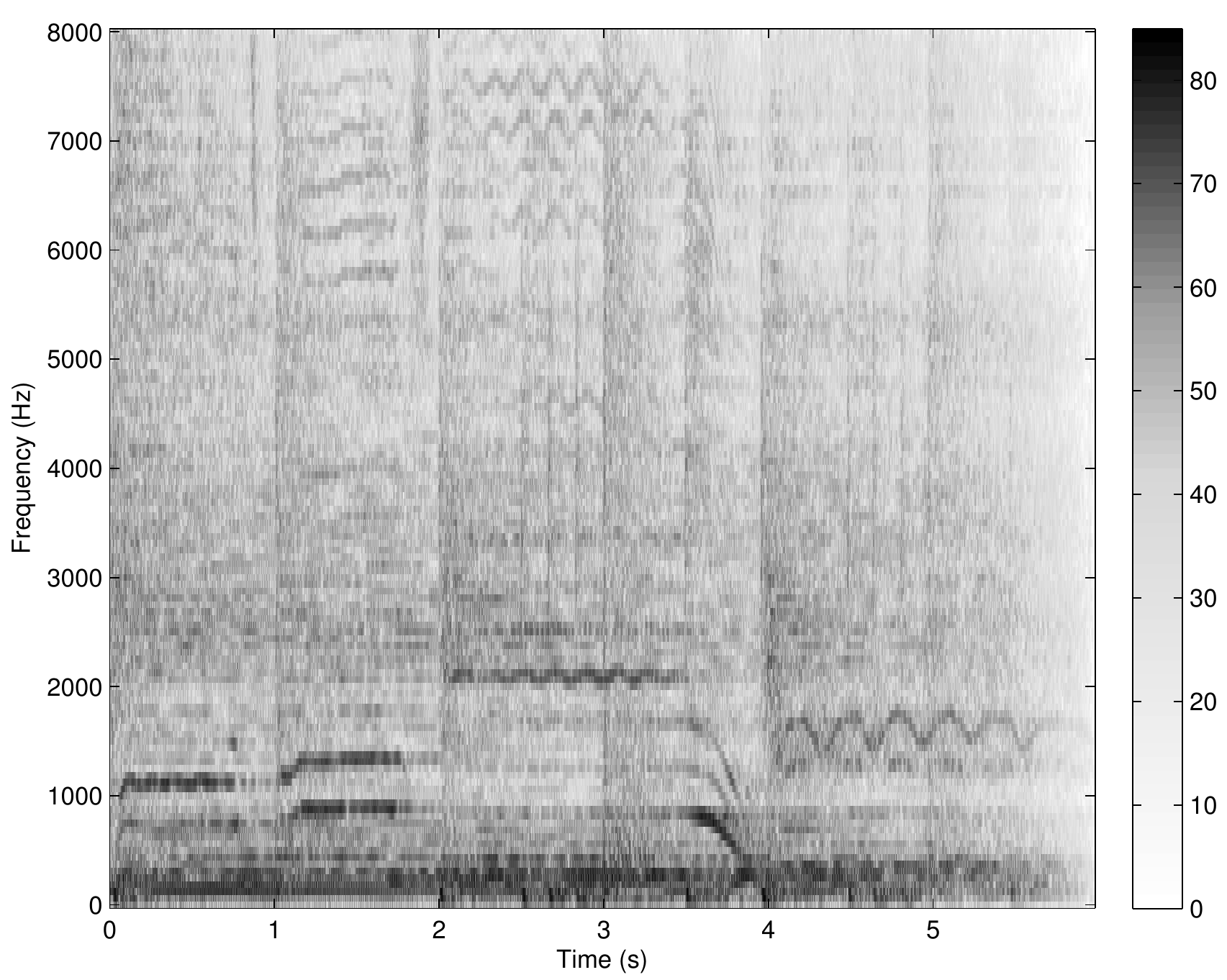}}
		 \subfloat[Binary mask]{
		 	\includegraphics[width=0.46\textwidth]{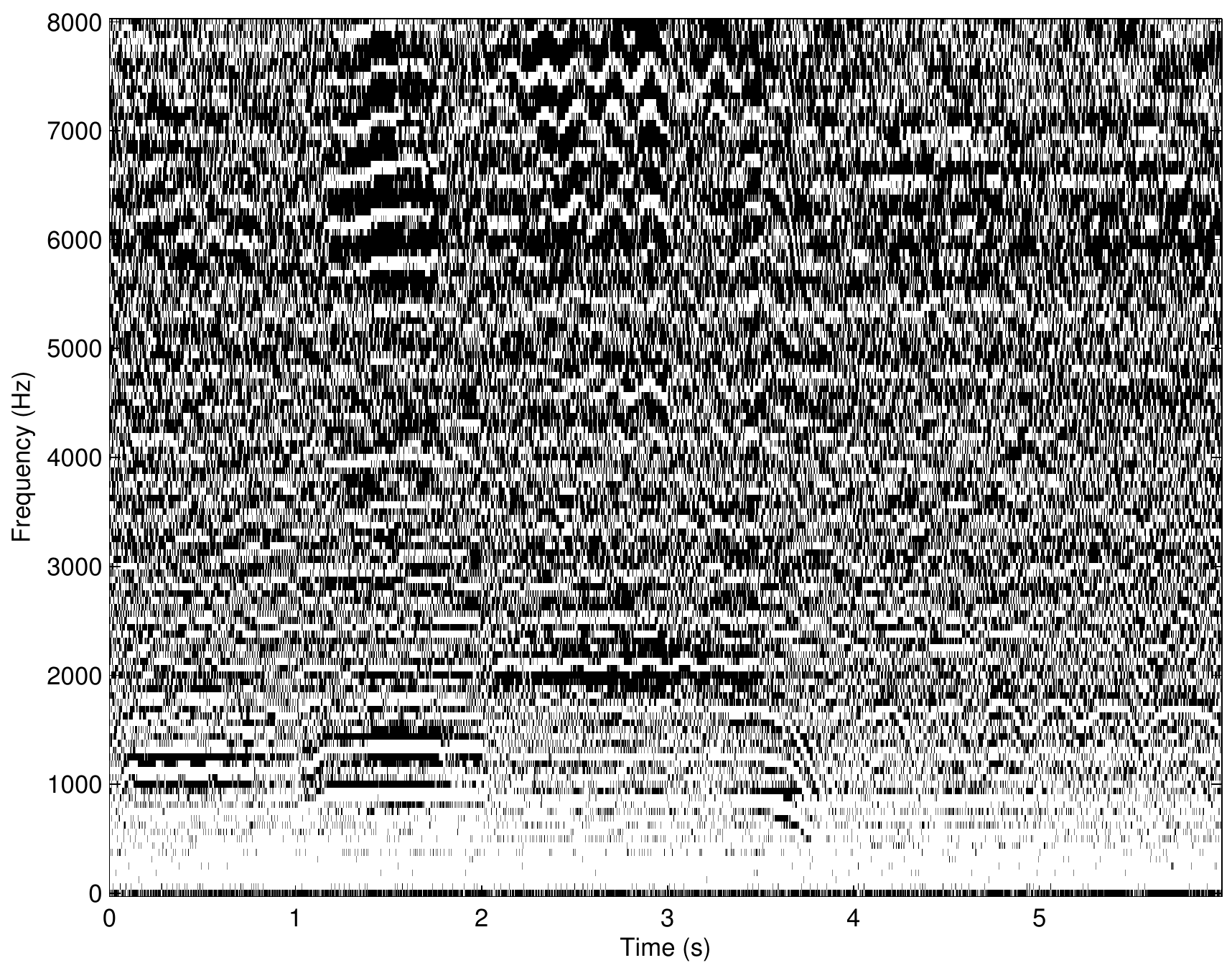}}\\[-8pt]
		 \subfloat[Masked Gabor transform]{
		 	\includegraphics[width=0.46\textwidth]{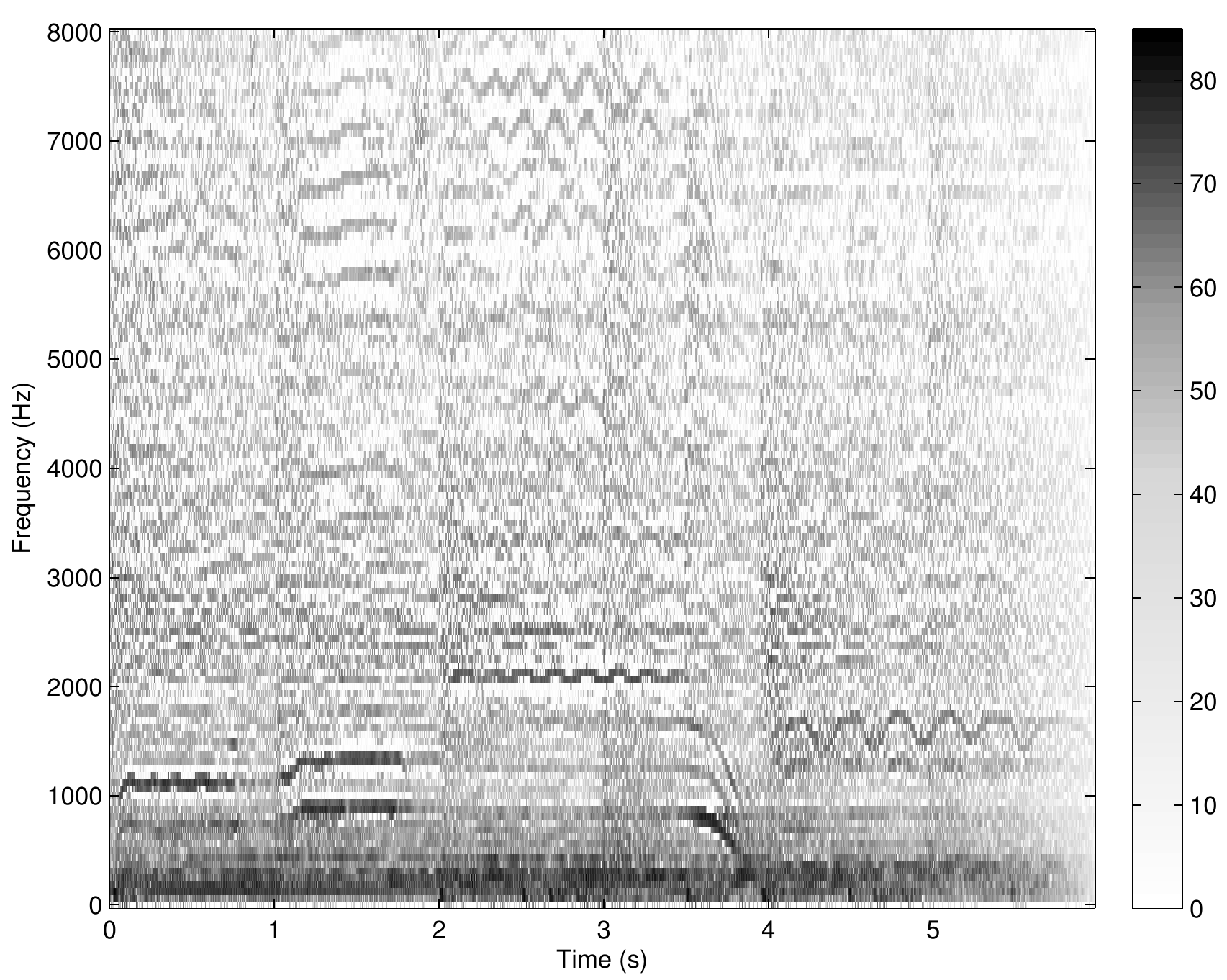}}
		 \subfloat[Irrelevance threshold]{
		 	\includegraphics[width=0.46\textwidth]{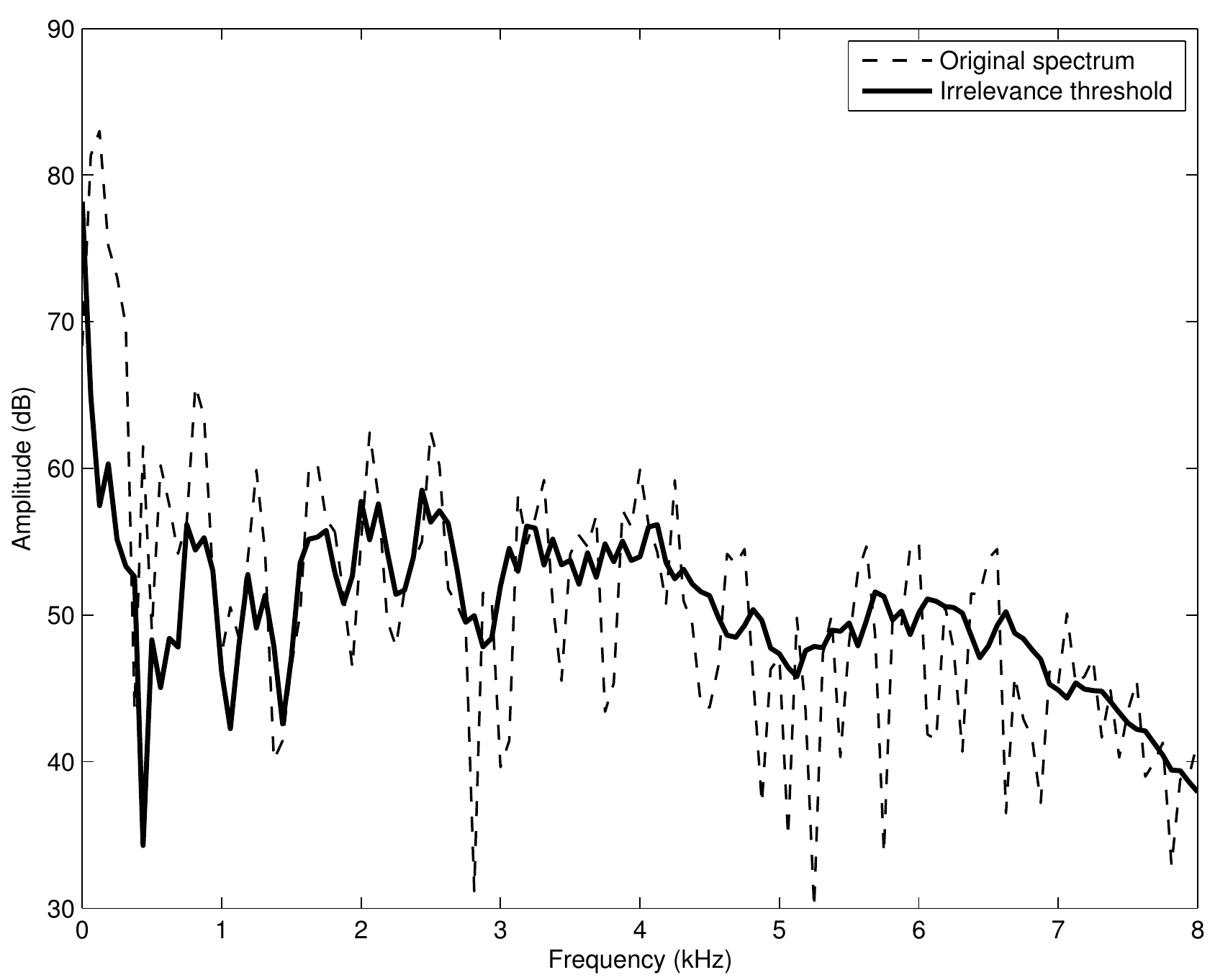}}
	  \end{center}
	  \caption{Example application of the irrelevance filter as implemented in \cite{Balazs:2010a} to a music signal (excerpt from the song ``Heart of Steel'' from Manowar). (a)~Squared magnitude of the Gabor transform (in~dB). (b)~Binary mask estimated from the irrelevance threshold. White = 1, black = 0. (c)~Squared magnitude (in~dB) of the masked Gabor transform, i.e. the result of the point-wise multiplication between the original transform and the binary mask. (d)~Amplitudes (in~dB) of the irrelevance threshold (bold straight line) and original spectrum (dashed line) at a given time slice.}
	  \label{fig:irrelfilter}  
\end{figure}

The irrelevance filter as depicted here has shown very promising results but the approach could be improved. Specifically, the main limitations of the algorithm are the fixed resolution in the Gabor transform and the use of a simple spectral masking model to predict masking in the time-frequency domain. Combining an invertible perceptually-motivated transform like the AUDlet FB (Sec.~\ref{sec:audlet0}) with a model of time-frequency masking (Sec.~\ref{sec:tfmask0}) is expected to improve performance of the filter. This is work in progress. Potential applications of perceptual sparsity include, for instance:
\begin{enumerate} 
\item Sound / Data Compression: 
 For applications where perception is relevant, there is no need to encode perceptually irrelevant information. Data that can not be heard should be simply omitted. A similar algorithm is for example used in the MP3 codec. If ``over-masking'' is used, i.e. the threshold is moved beyond the level of relevance, a higher compression rate can be reached \cite{Painter:2000a}. 
\item Sound Design: For the visualization of sounds the perceptually irrelevant part can be disregarded. This is for example used for car sound design \cite{Bezat:2007b}.  
\end{enumerate}

\section{Conclusion}
In this chapter, we have discussed some important concepts from hearing research and perceptual audio signal processing, such as auditory masking and auditory filter banks. Natural and important considerations served as a strong indicator that frame theory provides a solid foundation for the design of robust representations for perceptual signal analysis and processing. This connection was further reinforced by exposing the similarity between some concepts arising naturally in frame theory and signal processing, e.g. between frame multipliers and time-variant filters.
Finally, we have shown how frame theory can be used to analyze and implement invertible filter banks, in a quite general setting where previous synthesis methods might fail or be highly inefficient. The codes for Matlab/Octave to reproduce the results presented in Secs.~\ref{sec:frameth} and~\ref{sec:appli} in this chapter are available for download on the companion Webpage \url{https://www.kfs.oeaw.ac.at/frames_for_psychoacoustics}.

It is likely that readers of this contribution who are researchers in psychoacoustics or audio signal processing have already used frames without being aware of the fact. We hope that such readers will, to some extent, grasp the basic principles of the rich mathematical background provided by frame theory and its importance to fundamental issues of signal analysis and processing. With that knowledge, we believe, they will be able to better understand the signal analysis tools they use and might even be able to design new techniques that further elevate their research. 

On the other hand, researchers in applied mathematics or signal processing have been supplied with basic knowledge of some central psychoacoustics concepts. We hope that our short excursion piqued their interest and will serve as a starting point for applying their knowledge in the rich and various fields of psychoacoustics or perceptual signal processing.

\vspace{.1in}
{\bf Acknowledgments} The authors acknowledge support from the Austrian Science Fund (FWF) START-project FLAME ('Frames and Linear Operators for Acoustical Modeling and Parameter Estimation'; Y 551-N13) and the French-Austrian ANR-FWF project POTION (``Perceptual Optimization of Time-Frequency Representations and Audio Coding; I 1362-N30''). They thank B. Laback for discussions and W. Kreuzer for the help with a graphics software.

%
%
%


\end{document}